\documentclass[conference,letterpaper]{IEEEtran}

\usepackage{cite}
\usepackage{tikz}
\usetikzlibrary{shapes, arrows, patterns, positioning, calc}
\usepackage{graphicx,xcolor}
\usepackage{pgfplots}
\pgfplotsset{compat=1.18}

\usepackage{amsmath,amssymb,amsthm,amsfonts}
\usepackage{bbm}
\usepackage{bm}
\usepackage{dsfont}

\newtheorem{theorem}{Theorem}

\newtheorem{lemma}{Lemma}
\newtheorem{corollary}{Corollary}

\usepackage{mathtools}

\allowdisplaybreaks

\newcommand{\vol}{\operatorname{Vol}}

\newcommand{\ex}{\mathbb{E}}
\newcommand{\pr}{\mathbb{P}}

\title{On the Entropy of a Random Geometric Graph}

\author{\IEEEauthorblockN{Praneeth Kumar Vippathalla, Justin P. Coon, and Mihai-Alin Badiu}
 \IEEEauthorblockA{Department of Engineering Science\\ University of Oxford\\ OX1 3PJ Oxford, U.K. \\Email: \{praneeth.vippathalla,  justin.coon, mihai.badiu\}@eng.ox.ac.uk}
\thanks{This research was funded in whole or in part by the U. S. Army Research Laboratory and the U. S. Army Research Office (W911NF-22-1-0070 and W911NF-24-2-0102). For the purpose of Open Access, the authors have applied a CC BY public copyright licence to any Author Accepted Manuscript (AAM) version arising from this submission.}}

\begin{document}
\IEEEoverridecommandlockouts
\maketitle

\begin{abstract}
  In this paper, we study the entropy of a \emph{hard} random geometric graph (RGG), a commonly used model for spatial networks, where the connectivity is governed by the distances between the nodes. Formally, given a connection range $r$, a hard RGG $G_m$ on $m$ vertices is formed by drawing $m$ random points from a spatial domain, and then connecting any two points with an edge when they are within a distance $r$ from each other. The two domains we consider are the $d$-dimensional unit cube $[0,1]^d$ and the $d$-dimensional unit torus $\mathbb{T}^d$. We derive upper bounds on the entropy $H(G_m)$ for both these domains and for all possible values of $r$.  In a few cases, we obtain an exact asymptotic characterization of the entropy by proving a tight lower bound.  Our main results are that $H(G_m) \sim dm \log_2m$ for $0 < r \leq 1/4$ in the case of $\mathbb{T}^d$ and that the entropy of a one-dimensional RGG on $[0,1]$ behaves like $m\log m$ for all $0<r<1$. As a consequence, we can infer that the asymptotic structural entropy of an RGG on $\mathbb{T}^d$, which is the entropy of an unlabelled RGG, is $\Omega((d-1)m \log_2m)$ for $0 < r \leq 1/4$. For the rest of the cases, we conjecture that the entropy behaves asymptotically as the leading order terms of our derived upper bounds.
\end{abstract} 

\section{Introduction} \label{sec:introduction}
In many real-world networks, such as wireless, brain and social networks, connectivity is governed by the spatial separation of entities. For instance, devices in a wireless network share a communication link if they are geographically close enough to one another. A random network model that captures the fundamental aspects of these networks, and is widely used in their study, is the random geometric graph (RGG) model. Here, a random graph is generated by scattering $m$ nodes uniformly at random in a spatial domain and drawing an edge between any two nodes based on the distance between them. Over the past few decades, a variety of problems related to RGGs have been studied in detail, aiming to understand the statistics of graph properties \cite{PenroseBook,haenggi_2012}, the recoverability of the latent spatial embedding \cite{dani2024}, and their combinatorial properties \cite{mcdiarmid_diskgraph_2014}. In contrast, we are interested in the information-theoretic compression of these graphs.

In order to develop efficient storage methods for large graph datasets, researchers have recently focused on studying graph compression \cite{choi_structural_entropy,delgosha_universal_compression,nikpey_graph_side, martin_rate_distortion_sbm, bustin_lossy, mihai2021structural,Abbe2016GraphClusters, coon2018entropy,mihai_2018_dist_rgg,praneeth_side_2025,lossy_spatial_24}. The graph sources are typically non-standard and are unlike independent and identically distributed (i.i.d.) sources, which makes their study particularly challenging. Nevertheless,  many works have attempted to characterize Shannon entropy of various graph sources, because of its importance in lossless compression. For a random source $X$, it is well-known that the minimum expected length of a prefix-free code $l
^*(X)$ is within $1$ bit from the entropy $H(X)$\cite{cover_thomas}, i.e.,
$$H(X) \leq l^*(X) < H(X)+1,$$
and the minimum expected length of a one-to-one code $\ell^*(X)$ is close to the entropy \cite{alon_1994}, \cite[Lem.~1]{Abbe2016GraphClusters}:
$$H(X) - \log (H(X)+1) - \log e\leq \ell^*(X) < H(X).$$
Results on the characterization of entropy are known for Erd\H{o}s-R\'enyi (ER) random graphs and their structures\cite{choi_structural_entropy}, stochastic block model (SBM) graphs, \cite{Abbe2016GraphClusters}, random geometric graphs and their structures \cite{coon2018entropy, mihai_2018_dist_rgg, mihai2021structural}, among many others. 

The asymptotic behaviour of the entropy of a \emph{soft} RGG, where the connection rule is probabilistic, is well understood \cite{coon2018entropy}. On the other hand, the entropy of \emph{hard} RGG, where the nodes are connected by an edge iff they are within a connection range $r$, remains understudied. However, in a prior work \cite{mihai2021structural}, it was hinted that the asymptotic behavior of the entropy of a one-dimensional RGG on the unit interval $[0,1]$ is $m \log_2 m$ bits, where $m$ is the number of vertices of the graph. Later, \cite{Paton2022labelled_entropy} gave an argument for this entropic behavior in a special case of the one-dimensional model, i.e., when the connection range decays linearly to zero. In our current work, we will consider a hard RGG on two domains:  $d$-dimensional unit cube $[0,1]^d$ and the $d$-dimensional unit torus $\mathbb{T}^d$ with a fixed connection range $r$. Along with deriving asymptotic upper bounds on the entropy, we will proof under some conditions a matching lower bound, giving our main result
$$H(G_m) \sim dm\log_2 m,$$
when $r\leq \frac{1}{4}$ for an RGG $G_m$ with $m$ vertices on $\mathbb{T}^d$. 
Finally, we will use this result to infer the structural entropy of an RGG.

\section{Preliminaries}\label{sec:prelim}
Let $G_m=(V, E)$ be a graph with the vertex set ${V}=[m]:=\{1,2,\ldots,m\}$ and the edge set $E$ containing unordered pairs of vertices. We restrict to undirected and simple graphs, where simple means that there are no self-loops or multiple edges between any pair of two vertices.

For defining random geometric graphs, we will focus on two spatial domains: the $d$-dimensional unit cube $[0,1]^d$ and the $d$-dimensional unit torus $\mathbb{T}^d$. The unit torus is nothing but $[0,1]^d$ with the boundaries wrapped around. In the unit cube, the distance between two points is given by the Euclidean distance, whereas in the unit torus, it is given by the toroidal distance\footnote{In fact, the minimization in the definition is equivalent to the minimization over $z \in \left\{-1, 0, 1\right\}^d$.} \cite{PenroseBook, haenggi_2012}, which is defined as  $$d_t(x,y):=\min\{\lVert x+z-y \rVert: z \in \mathbb{Z}^d\} \quad \text{for} \quad x, y\in \mathbb{T}^d$$
with $\lVert \cdot \rVert$ being the Euclidean norm.

For a fixed connection range $r \geq 0$, we say that a graph $G_m$ is a \emph{hard geometric graph} on $[0,1]^d$ (resp. $\mathbb{T}^d$) if there exist  $m$ points $x_1, x_2, \ldots, x_m \in [0,1]^d$ (resp. $\mathbb{T}^d$) such that for any $u,v \in V$, $||x_u -x_v|| \leq r$ (resp. $d_t(x_u,y_v) \leq r$ ) if and only if $u\sim v$, i.e., $u$ is adjacent to $v$ in $G_m$. The toroidal metric overcomes some of problems posed by the boundary effects of the unit cube, i.e., the  local connectivity properties of the points near the boundary of $[0,1]^d$ is different from those of the points well within it.

In the {hard} random geometric graph model, we draw $m$ points $X_1, X_2, \ldots, X_m$ independently with uniform distribution on $[0,1]^d$ (resp. $\mathbb{T}^d$) and form the corresponding geometric graph $G_m$ on the vertex set $[m]$. This graph generation process induces a probability distribution $P_{G_m}$ on the set of all graphs. It is useful to think of $G_m$ in terms of its adjacency matrix representation $(E_{i,j})_{i,j \in [m]}$, where $E_{i,j}=1$ if $i$ and $j$ are adjacent in $G_m$, and $E_{i,j}=0$ otherwise. As the graphs are simple and undirected, this matrix is symmetric with zeros on the diagonal. It means that $G_m$ can  equivalently be represented by $\left\{E_{i,j}: 1\leq i < j \leq m \right\}$. Hence, the probability of a random geometric graph $G_m$ is completely specified by the joint distribution of the collection $\left\{E_{i,j}: 1\leq i < j \leq m \right\}$, and vice versa. In the case of  the Erd\H{o}s-R\'enyi graph model, where we add an edge between two vertices with some probability independently of the rest of the edges, this collection is mutually independent. However, in the random geometric graph model, this collection is correlated because of the underlying geometry. For example, if $u \sim v$ and $u \sim w$, then more often than not we see an edge between $v$ and $w$.

\subsection{Entropy of $G_m$}
Let $\mathcal{G}_m \subseteq 2^{^{\left[\binom{m}{2}\right]}}$ denote the set of all geometric graphs on $m$ vertices with the connection range $r$, i.e., $\mathcal{G}_m$ is the support of the probability distribution $P_{G_m}$. The entropy\footnote{With a logarithm base $2$, we measure the entropy in bits. In the rest of the note, though we omit the usage of base $2$, the entropy is still being measured in bits.} of a random geometric graph $G_m$ is defined as 
\begin{align*}
    H(G_m) \triangleq -\sum_{g \in \mathcal{G}_m} P_{G_m}(g) \log_2 P_{G_m}(g).
\end{align*}
As $G_m$ can be uniquely identified with the collection of random variables $\left\{E_{i,j}: i < j \right\}$, we have $H(G_m) = H(\left\{E_{i,j}: i < j \right\})$. By using various properties of entropy, the following upper and lower bounds were noted in the work of \cite{coon2018entropy}: 
\begin{align*}
    H(G_m|X_1, \ldots, X_m) \leq H(G_m) \leq \sum_{ i < j }H(E_{i,j}),
\end{align*}
where $X_1, \ldots, X_m \in [0,1]^d$ are the random variables corresponding to the locations of the vertices. 
Though these bounds are useful in determining the behavior of  $H(G_m)$ for soft random geometric graphs, they fall short in the case of hard RGGs, which is due to the fact that the lower bound evaluates to zero:
\begin{align}\label{ineq:general_bounds}
    0 \leq H(G_m) \leq \binom{m}{2}h_2(p_r),
\end{align}
where $h_2(x):= -x\log_2 x -(1-x)\log_2(1-x)$ is the binary entropy function, and $p_r$ is the probability that two random nodes are within a distance of $r$ from each other.  For a fixed $r$, the upper bound in \eqref{ineq:general_bounds} behaves like $m^2/2$ and it is tight when the $E_{i,j}$'s are mutually independent. However, since $E_{i,j}$'s are correlated in an RGG, it is reasonable to say that $m^2/2$ might not be the right behavior. In the subsequent sections, we will try to improve these bounds to characterize the exact behavior of $H(G_m)$.

\section{Entropy of an RGG}\label{sec:bounds}
\subsection{$d$-dimensional unit cube $[0,1]^d$}
In this section, we present the results on the entropy of an RGG on the $d$-dimension unit cube $[0,1]^d$. First, we will give an upper bound for all fixed connection ranges $0 < r < \sqrt{d}$. Observe that when $r$ equals $0$ or $1$, the resulting graph is always an empty graph or a complete graph, respectively, whose entropy is trivially zero. Hence, we exclude those two extreme connection ranges.
\begin{theorem}\label{thm:ent_rgg_upper_euclidean}
An upper bound on the entropy of an RGG $G_m$ on the $d$-dimensional unit cube $[0,1]^d$ is given by
     \begin{align*}
        &H(G_m) \nonumber \\
        &\leq \begin{cases}
            dm\log m + o(m \log m) & \text{ if } 0 < r \leq \frac{\sqrt{d}}{2},\\
            [1-\beta(r)]dm\log m + o(m \log m) & \text{ if }  \frac{\sqrt{d}}{2} \leq r < \sqrt{d}.
        \end{cases}
    \end{align*}   
where $\beta(r)$ is the volume of the ball $B\left((1/2,1/2,\ldots,1/2); r-\sqrt{d}/2\right) \cap [0,1]^d$. 
\end{theorem}
\begin{proof}
    See Section~\ref{proof:thm:ent_rgg_upper_euclidean}.
\end{proof}

Observe that the leading order term of the upper bound in Thm.~\ref{thm:ent_rgg_upper_euclidean} is less than $dm \log m$ and it decreases with $r$ beyond $\frac{\sqrt{d}}{2}$. This refinement is possible due to the fact that when the connection range is larger than $\frac{\sqrt{d}}{2}$, all the random points that lie within  $B\left((1/2,1/2,\ldots,1/2); r-\sqrt{d}/2\right) \cap [0,1]^d$ are connected by edges to the rest of the points, and such connections are significant in a random graph, resulting in entropy loss. The proof of Thm.~\ref{thm:ent_rgg_upper_euclidean} involves bounding the cardinality of the ensemble of RGGs for all $0<r< \sqrt{d}$, following the argument of \cite{mcdiarmid_diskgraph_2014} that uses the Warren's theorem (Thm.~\ref{thm:warren}), a result on the number of possible sign patterns of a collection of polynomials. Combining this cardinality bound with entropic inequalities, we obtain a refined bound for  $\frac{\sqrt{d}}{2} \leq r < \sqrt{d}$.

Due to the boundary effects of the domain, the analysis for a lower bound is technically challenging. However, we believe that the leading order terms are the correct asymptotic behavior of the entropy. To confirm this, we considered the case of $d=1$ and proved a matching lower bound, which is presented in the following theorem. 

\begin{theorem}\label{thm:ent_rgg_lower_one_d}
The entropy of an RGG $G_m$ on $[0,1]$ is lower bounded as follows:
     \begin{align}
        H(G_m) \geq \begin{cases}
            m\log m - o(m \log m) & \text{ if } 0 < r \leq \frac{1}{2},\\
            2(1-r)m\log m - o(m \log m) & \text{ if }  \frac{1}{2} \leq r < 1.
            \end{cases}
    \end{align}   
\end{theorem}
\begin{proof}
    See Section~\ref{proof:thm:ent_rgg_lower_one_d}.
\end{proof}

The proof of Thm.~\ref{thm:ent_rgg_lower_one_d} relies on the order statistics of $m$ points that are uniformly and independently drawn from $[0,1]$. By combining Thms.~\ref{thm:ent_rgg_upper_euclidean} and \ref{thm:ent_rgg_lower_one_d}, and noting that $\beta(r)=2r-1$ for $\frac{1}{2}\leq r <1$, we have the following result. The limit of the entropy normalized by $m\log m$ as a function of the connection range $r$ is plotted in Fig.~\ref{fig:my_label}.

\begin{theorem}\label{thm:asymp:cube}
    The asymptotic characterization of the entropy of an RGG $G_m$ on $[0,1]$ is given by
    \begin{align}
        H(G_m) \sim \begin{cases}
            m\log m & \text{ if } 0 < r \leq \frac{1}{2},\\
            2(1-r)m\log m & \text{ if }  \frac{1}{2} \leq r < 1.
            \end{cases}
    \end{align}
\end{theorem}

\begin{figure}[h]
\hspace*{0.7cm}
       \resizebox{.6\columnwidth}{!}{\begin{tikzpicture}
\begin{axis}[
    width=8cm,
    height=6.5cm,
    xmin=0, xmax=1.05,
    ymin=0, ymax=1.1,
    axis lines=box,
    xlabel={$r$},
    ylabel={$\displaystyle \lim_{m\to\infty}\frac{H(G_m)}{m\log m}$},
      ylabel style={
        rotate=270,
        anchor=center,
        at={(axis description cs:-0.3,0.5)},
        font=\small
    },
    xtick={0,0.2,0.4,0.6,0.8,1},
    ytick={0,0.2,0.4,0.6,0.8,1},
    tick style={black},
    label style={font=\small},
    legend style={draw=none},
]

\addplot[
    blue,
    thick,
    domain=0:0.5
] {1};

\addplot[
    blue,
    thick,
    domain=0.5:1
] {2*(1-x)};

\addplot[
    dashed,
    black
] coordinates {(0.5,0) (0.5,1)};

\node at (axis cs:0.25,1.05) {$1$};
\node at (axis cs:0.82,0.6) {$2(1-r)$};


\end{axis}
\end{tikzpicture}}
       \caption{Limit of $H(G_m)/m\log m$ of a one-dimensional RGG}
       \label{fig:my_label}
   \end{figure}
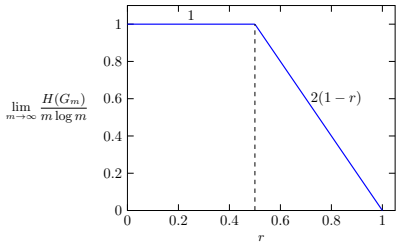

\subsection{$d$-dimensional unit torus $\mathbb{T}^d$}
Some of the technical difficulties posed by boundary effects of the domain $[0, 1]^d$ can be avoided by considering the unit torus $\mathbb{T}^d$, where the local connectivity is statistically the same at every point.  Note that the connection range on $\mathbb{T}^d$ cannot\footnote{For any $x, y\in \mathbb{T}^d$, $d_t(x,y) \leq \frac{\sqrt{d}}{2}$ because $d_t(x,y)^2= \sum \limits_{i=1}^{d} \min \limits_{z_i \in \{-1,0,1\}} |x_i+z_i-y_i |^2 \leq \sum \limits_{i=1}^{d} \frac{1}{4}$.} be more than $\frac{\sqrt{d}}{2}$. The following theorem presents an entropy upper bound in the case of a unit torus.

\begin{theorem}\label{thm:ent_rgg_upper_torus}
The entropy of an RGG $G_m$ on $\mathbb{T}^d$ is upper bounded as follows:
     \begin{align}
        H(G_m) \leq 
            dm\log m + o(m \log m)
    \end{align} 
for $0 < r < \frac{\sqrt{d}}{2}$.
\end{theorem}
\begin{proof}
    See Section~\ref{proof:thm:ent_rgg_upper_torus}.
\end{proof}

Even on a unit torus, the leading order term of the upper bound is still $dm\log m$ for the entire connection range $0 < r < \frac{\sqrt{d}}{2}$. The basic proof idea of Thm.~\ref{thm:ent_rgg_upper_torus} is similar to that of Thm.~\ref{thm:ent_rgg_upper_euclidean}--- bounding the cardinality of the ensemble of RGGs using the Warren's theorem (Thm.~\ref{thm:warren}). However, we choose a different collection of polynomials based on the toroidal metric. 

A tight lower bound is presented in the next theorem, which holds for the connection range values $r\leq \frac{1}{4}$. This condition is a mere technical choice to simplify the analysis. However, we believe that the same result is valid over the whole interval  $0 < r < \frac{\sqrt{d}}{2}$. 
\begin{theorem}\label{thm:ent_rgg_lower_torus}
The entropy of an RGG $G_m$ on $\mathbb{T}^d$ is lower bounded as follows:
     \begin{align}
        H(G_m) \geq  dm\log m - o(m \log m)
    \end{align}   
for $r\leq \frac{1}{4}$. 
\end{theorem}
\begin{proof}
    See Section~\ref{proof:thm:ent_rgg_lower_torus_boolean} and Section~\ref{proof:thm:ent_rgg_lower_torus} for two different proofs.
\end{proof}

One proof relies on a result on
the Crofton cell properties \cite{Richey_Sarkar_2022} for a particular Boolean model whose focus is on the intersection of randomly placed spheres. Another proof involves bounding the volume of the crescents formed from the intersection of two spheres with random centers.

Now our main result can be stated by combining Theorem~\ref{thm:ent_rgg_upper_torus} and \ref{thm:ent_rgg_lower_torus}. The entropy scales exactly like $dm\log m$, which is linear in the dimension $d$ and independent of the connection range $r$.

\begin{theorem}\label{thm:asymp:torus}
    The asymptotic characterization of the entropy of an RGG $G_m$ on $\mathbb{T}^d$ is given by
    \begin{align}
        H(G_m) \sim dm \log m
    \end{align}
    for $r \leq  \frac{1}{4}$.
\end{theorem}

\subsection{Structural Entropy of an RGG}
The entropy characterization result (Thm.~\ref{thm:asymp:torus}) can be used to deduce how small the structural entropy of an RGG can be. A structure \cite{choi_structural_entropy} of a graph is its unlabelled version. With applications to compression, the structural entropy of various random graphs was studied \cite{choi_structural_entropy,luczak_pag_asym,mihai2021structural}. In the work \cite{mihai2021structural} on one-dimensional RGGs, it was established that the structural entropy $H(S_m)$, where $S_m$ is the structure of an RGG $G_m$, satisfies the relation $\log e (1-r) m \leq H(S_m) \leq 2(1-r) m $ for $0 < r < 1$. By combining this with the result of Thm.~\ref{thm:asymp:cube}, we can conclude that for a one-dimensional RGG, the randomness of the graph is dominated by the randomness in the labels. To see this, consider the expansion of the entropy of a general RGG in terms of the structural entropy:
\begin{align}\label{eq:decomp}
    H(G_m) = H(G_m, S_m) &= H(S_m) + H(G_m|S_m)
\end{align}
which follows from the fact that $S_m$ is completely determined by $G_m$. Notice that when conditioned on a structure $S_m$, the randomness of the graph $G_m$ lies in its labeling. Hence, $H(G_m|S_m)$ represent the average conditional randomness in the labeling. It follows from \eqref{eq:decomp} that the entropy of this labeling must  be $\Theta(m\log m)$ for a one-dimensional RGG as $H(G_m)= \Theta(m\log m)$ and $H(S_m)=\Theta(m)$, dominating the entropy of the graph. However, this is not the case when $d>1$. Here, the structural entropy dominates as shown in the following result.
\begin{corollary}
    \label{cor:asymp:struct:torus}
    The structural entropy of an RGG $G_m$ on $\mathbb{T}^d$ is given by
    \begin{align}
        H(S_m) = \Omega ((d-1)m \log m)
    \end{align}
    for $r \leq  \frac{1}{4}$.
\end{corollary}
\begin{proof}
    
Since  $H(G_m|S_m)$  is upper bounded by the log of cardinality of the all possible labelings, which is $m!$, we have from \eqref{eq:decomp} 
$H(S_m)\geq H(G_m)-H(G_m|S_m) \geq H(G_m) -\log(m!)$.
The result immediately follows from the Stirling's approximation and Thm.~\ref{thm:asymp:torus}.

\end{proof}

\section{Proofs}\label{sec:proofs}
\subsection{Proof of Theorem~\ref{thm:ent_rgg_upper_euclidean}: Upper bound in the case of $[0,1]^d$}
\label{proof:thm:ent_rgg_upper_euclidean}
The first bound in the theorem  follows from the work of McDiarmid and Müller \cite{mcdiarmid_diskgraph_2014} on disk-intersection graphs, where Warren's theorem \cite{warren1968} (stated below for convenience) was used to find an upper bound on the cardinality of the set of possible graphs\footnote{Note that not every graph is a geometric graph; for example, it is impossible to realize the complete bipartite graph $K_{1,6}$ (a star graph)  on  $[0,1]^2$ \cite[Ch. 3]{PenroseBook}. Therefore, the number of possible geometric graphs is strictly less than $2^{\binom{m}{2}}$, the total number of possible graphs with $m$ vertices.}. For the sake of completeness, we will present that argument applied to random geometric graphs.

Recall that $\mathcal{G}_m$ denotes the set of all hard RGGs with a fixed connection radius $r$. In a hard RGG, two vertices $u,v \in V$ are connected by an edge iff $\lVert x_u -x_v \rVert \leq r$. Now we consider the degree-$2$ polynomials of the form $\lVert X_u -X_v\rVert^2 - r^2$ for each pair $u, v \in V$. For a fixed pair of vertices $u$ and $v$, if the sign of the polynomial  $\lVert X_u -X_v\rVert^2 - r^2$ evaluated at $X_u=x_u$ and $X_v=x_v$ is negative, then there is an edge in the graph $G_m$ between $u$ and $v$; otherwise, there is no edge.

By arranging the signs of the evaluations of all these polynomials (in a fixed order) at a realization of the node locations $X_1=x_1, \dots X_m=x_m$, we get a sign pattern. Let $\mathcal{S}_m \subseteq \left\{+, -\right\}^{\binom{m}{2}}$ be the collection of all such sign patterns. In other words,  $\mathcal{S}_m$ is the set containing all the sign patterns corresponding to the points $(x_1, \ldots, x_m) \in [0,1]^{dm}$. As each sign pattern is uniquely identified with an RGG, and vice versa, the map from $\mathcal{S}_m$ to $\mathcal{G}_m$ is bijective. Therefore, $|\mathcal{G}_m| = |\mathcal{S}_m|$. We can now bound $|\mathcal{S}_m|$ using the following theorem.

\begin{theorem}[Warren \cite{warren1968,mcdiarmid_diskgraph_2014}]\label{thm:warren}
    For polynomials $Q_1, Q_2, \ldots, Q_u$ with at most degree $k$ in real variables $z_1, \ldots, z_t$, the number of the sign patterns $(\operatorname{sign}(Q_1), \ldots, \operatorname{sign}(Q_u)) \in \{-,+\}^u$ evaluated over $\mathbb{R}^t \setminus \cup_{i=1}^u\{(z_1, \ldots, z_t): Q_i(z_1, \ldots, z_t)=0\}$ is bounded above by $\left(\frac{4eku}{t}\right)^t$.    
\end{theorem}

By applying Theorem~\ref{thm:warren} with $u=\binom{m}{2}$, $t=dm$ and $k=2$ for the above degree-$2$ polynomials, we get
 \begin{align*}
     \lvert\mathcal{G}_m\rvert=\lvert\mathcal{S}_m\rvert \leq \left(\frac{4e\cdot 2 \cdot \binom{m}{2}}{md}\right)^{dm} \leq m^{dm} \cdot \left(\frac{4e}{d}\right)^{dm},
 \end{align*}
 which is indeed independent of the connection radius $r$. This implies that 
 \begin{align}
       H(G_m) \leq \log \lvert \mathcal{G}_m \rvert \leq dm\log_2 m + dm \log_2 C, \label{ineq:card_bound_rgg_cube}
\end{align}
where  $C= \left(\frac{4e}{d}\right)^{d}$. This bound is true for the whole range $0 < r < \sqrt{d}$, proving the first part of the theorem. However, we can improve this bound if $\frac{\sqrt{d}}{2} \leq r < \sqrt{d}$, which is presented below.

Let us denote the ball $B\left((1/2,1/2,\ldots,1/2); r-\frac{\sqrt{d}}{2}\right) \cap [0,1]^d$ around the center of $[0,1]^d$ of radius $r-\sqrt{d}/2$ by $B\left(r-\frac{\sqrt{d}}{2}\right)$ and its  volume $\vol\left(B\left(r-\frac{\sqrt{d}}{2}\right)\right)$ by $\beta(r)$. Define $S\triangleq\left\{i \in [m]: X_i \in B\left(r-\sqrt{d}/2\right)\right\}$. Note that for any $i \in S$, $\lVert X_v -X_i \rVert \leq r$ for all $v \in [m]\setminus\{i\}$ because every point $y \in [0,1]^d$ is within a distance\footnote{By triangle inequality, $\lVert y - x\rVert \leq \lVert y - (1/2,1/2,\ldots,1/2)\rVert + \lVert (1/2,1/2,\ldots,1/2) - x\rVert \leq \frac{\sqrt{d}}{2} + \left(r-\frac{\sqrt{d}}{2}\right) = r$} $r$ from every $x \in B\left(r-\sqrt{d}/2\right)$. Therefore, any edge random variable corresponding to a vertex $i \in S$ is $1$. We can think\footnote{We use the notation $G[U]$ to denote the vertex-induced subgraph of $G_m:= G[V]$, where $V=[m]$, containing only the vertices in $U \subseteq V$.} of $G[S]$, the graph with only vertices in $S$ as the \emph{core} of the whole graph $G[V]=G_m$. The graph $G[S]$ is complete with significant number of vertices in it, and it is connected to all the other vertices within $G[V]$. If we know the set $S$, the only uncertainty in the graph $G[V]$ is due to $G[V \setminus S]$. Using this fact we can bound the entropy:
\begin{align}
    H(&G_m)= H(G[V])\leq H(G[V], S) \nonumber \\
    & = H(S)+ H(G[V] \mid S) \nonumber \\
    & = H(S) + H(G[V\setminus S]|S)\label{eq:up_ent_1}\\
    & \leq \log 2^m + H(G[V\setminus S]|S)\label{eq:up_ent_2}\\
    & = m + \mathbb{E}_S\left[H(G[V\setminus S]|S)\right]\nonumber\\
    & \leq m + \mathbb{E}_S\left[d |V\setminus S| \log |V\setminus S| + |V\setminus S|\log C\right]\label{eq:up_ent_3}\\
    &\leq m +  d \cdot \mathbb{E}_S\left[|V\setminus S|\right] \log m + \mathbb{E}_S\left[|V\setminus S|\right]\log C \label{eq:up_ent_4}\\
    & = m + (1-\beta(r))m \log m + (1-\beta(r))m \log C,\label{eq:up_ent_5}
\end{align}
where \eqref{eq:up_ent_1} relies on the fact that all the edge random variables in the graph $G[V]$ except for those in $G[V\setminus S]$ are $1$; as the cardinality of all possible sets $S$ is $2^m$, \eqref{eq:up_ent_2} follows; in \eqref{eq:up_ent_3} and \eqref{eq:up_ent_4}, we use the upper bound on the entropy of a random geometric graph in terms of the cardinality of the number graphs \eqref{ineq:card_bound_rgg_cube} with $|V\setminus S|$ vertices, i.e., \begin{align*}
H(G[V\setminus S]|S=s)&\leq \log \lvert \mathcal{G}_{|V\setminus s|}\rvert\nonumber \\ & \leq d |V\setminus s| \log |V\setminus s| + |V\setminus s|\log C\\& \leq  d |V\setminus s| \log m + |V\setminus s|\log C ;
\end{align*}
in \eqref{eq:up_ent_5} is due to the fact that $\mathbb{E}_S\left[|V\setminus S|\right]= m- \mathbb{E}_S\left[|S|\right]= m - m \beta(r)=m(1-\beta(r))$.
This shows that 

\begin{align*}
    H(G_m)\leq  [1-\beta(r)]m \log m + [1+(1-\beta(r))\log C]m
\end{align*}
for $\frac{\sqrt{d}}{2} \leq r < \sqrt{d}$, completing the proof of the theorem.

\subsection{Proof of Theorem~\ref{thm:ent_rgg_upper_torus}: Upper bound in the case of $\mathbb{T}^d$}\label{proof:thm:ent_rgg_upper_torus}
In a random geometric graph on $\mathbb{T}^d$, two vertices $u$ and $v$ are not connected by an edge iff the underlying node locations $x_u$ and $x_v$ satisfy $d_t(x_u, x_v) > r$. Note that the minimization in the toroidal metric $d_t(x,y):=\min\{\lVert x+z-y \rVert: z \in \mathbb{Z}^d\}$ is equivalent to the minimization over $z \in \left\{-1, 0, 1\right\}^d$. Therefore, the necessary and sufficient condition for the absence of an edge becomes $\lVert x_u + z -x_v\rVert > r$ for all $z \in \left\{-1, 0, 1\right\}^d$. 

Now we consider the degree-$2$ polynomials of the form $\lVert X_u + z -X_v\rVert^2 - r^2$ for each $z \in \{-1, 0, 1\}^d$ and each pair $u, v \in V$. For a fixed pair of vertices $u$ and $v$, if the signs of the polynomials  $\lVert X_u + z -X_v\rVert^2 - r^2$,  $z \in \{-1, 0, 1\}^d$, evaluated at $X_u=x_u$ and $X_v=x_v$ are all positive, then there is no edge in the graph $G_m$ between $u$ and $v$. On the other hand, there will be an edge if there is at least one polynomial evaluation with a negative sign.

By arranging the signs of the evaluations of all these polynomials (in a fixed order) at a realization of the node locations $X_1=x_1, \dots X_m=x_m$, we get a sign pattern. Let $\mathcal{S}_m \subseteq \left\{+, -\right\}^{3^d\binom{m}{2}}$ be the collection of all such sign patterns. In other words,  $\mathcal{S}_m$ is the set containing all the sign patterns corresponding to the points $(x_1, \ldots, x_m) \in [0,1]^{dm}$. Note that each sign pattern in $\mathcal{S}_m$ correspond to a random geometric graph. On the other hand, for each random geometric graph, there is at least one sign pattern in $\mathcal{S}_m$ corresponding to it. This is due to the fact that different realizations of node locations would produce a single random geometric graph, and the sign patterns associated with these realizations could be different because the necessary and sufficient condition for an edge to be present is to have at least one negative sign corresponding to a $z \in \{-1, 0, 1\}^d$ in the relevant part of the sign pattern. In other words, the map from $\mathcal{S}_m$ to $\mathcal{G}_m$ is surjective, i.e., $|\mathcal{G}_m|\leq |\mathcal{S}_m|$. We can now bound $|\mathcal{S}_m|$ by applying Warren's theorem (Thm.~\ref{thm:warren}) with $u=3^d \cdot \binom{m}{2}$, $t=dm$ and $k=2$  to the above degree-$2$ polynomials:
\begin{align}\label{ineq:card_bound_rgg}
     \lvert\mathcal{G}_m\rvert \leq \lvert\mathcal{S}_m\rvert \leq \left(\frac{4e\cdot 2 \cdot 3^d \cdot \binom{m}{2}}{md}\right)^{dm}  \leq m^{dm} \cdot C^{dm},
 \end{align}
 where $C=\frac{4e\cdot 3^d}{d}$.  
This implies that 
 \begin{align}
       H(G_m) \leq \log \lvert \mathcal{G}_m \rvert \leq dm\log_2 m + dm \log_2 C.
\end{align}
This bound is true for $0 < r < \frac{\sqrt{d}}{2}$, proving the theorem. 

\subsection{Proof of Theorem~\ref{thm:ent_rgg_lower_torus}: Lower Bound in the case of $\mathbb{T}^d$} \label{proof:thm:ent_rgg_lower_torus_boolean} 

We begin with the series of equalities involving the differential entropies\footnote{The differential entropy of an absolutely continuous random variable $X$ with probability density function $f_X(x)$ is defined as $h(X):= -\int f_x(x) \log f_X(x) dx.$}:
\begin{align}
    H(G_m)  &= I(G_m;X_1, \ldots, X_m) \label{eq:diff_ent_1}\\
            &= h(X_1, \ldots, X_m) - h(X_1, \ldots, X_m \mid G_m) \nonumber\\
            &= - h(X_1, \ldots, X_m \mid G_m) \label{eq:diff_ent_2}\\
            &= - \sum_{k=1}^m h(X_k \mid X_1, \ldots, X_{k-1}, G_m),\nonumber
\end{align}
where \eqref{eq:diff_ent_1} results from the fact that $H(G_m \mid X_1, \ldots, X_m) = 0$, and $\eqref{eq:diff_ent_2}$ follows because the $X_i$'s are uniformly distributed and $\vol([0,1]^d) = 1$, therefore $h(X_1, \ldots, X_m)=0$.  We now obtain a bound by ignoring all connections in $G_m$ that do not relate directly to node $k$:
\begin{align}
    H(G_m) &\geq - \sum_{k=1}^m h(X_k \mid X_1, \ldots, X_{k-1}, E_{1k}, E_{2k}, \ldots, E_{k-1,k})\nonumber\\
    & = - \sum_{k=1}^m h(X_k \mid X_1, \ldots, X_{k-1}, L_k)\label{eq: ent1}
\end{align}
where the first inequality is a consequence of the property that conditioning reduces entropy, and $L_k\triangleq\{i \in [k-1]: E_{ik}=1\}$, which is a random subset of $\{1, \ldots, k-1\}$. Furthermore, for fixed $X^{k-1} = x^{k-1}$ and $L_k = \ell_k$, the differential entropy is just the logarithm
of the volume of the region $\mathcal{R}(X^{k-1},L_k)$ in $\mathbb{T}^d$ carved out by the intersection of the balls centered at $x_i$, $i \in L_k$, excluding the portion covered by the balls centered at $x_i$, $i \in L_k^c$.  From this reasoning, we have
\begin{align*}
    h(X_k \mid X_1, \ldots, X_{k-1}, L_k) 
        &= \ex \left[\log \left[\vol \left( \mathcal{R}(X^{k-1},L_k) \right)\right]\right] \\
        &\leq \log \ex \left[\vol \left( \mathcal{R}(X^{k-1},L_k) \right)\right],
\end{align*}
where the second line follows from Jensen's inequality, and
\[
    \mathcal{R}(X^{k-1},L_k) = \bigg[\bigcap_{i \in L_k} \mathcal{B}_{X_i}(r)\bigg]\bigcap\bigg[\bigcap_{i \in L_k^c} \mathcal{B}^c_{X_i}(r)\bigg].
\]
Let $\mathcal{I}(X^{k-1},L_k) = \bigcap \limits_{i \in L_k} \mathcal{B}_{X_i}(r)$ and take $\mathcal{I}(X^{k-1},\emptyset) = \mathbb{T}^d$ by definition.  Note that $$\vol \left( \mathcal{R}(X^{k-1},L_k) \right) \leq \vol \left( \mathcal{I}(X^{k-1},L_k) \right)$$ for $L_k \subseteq [k-1]$, which immediately gives 
\begin{align}
    h(X_k \mid X_1, \ldots, X_{k-1}, L_k) &\leq \log \ex \left[\vol \left( \mathcal{I}(X^{k-1},L_k) \right)\right].\label{eq:ent2}
\end{align}
Let $L$ denote the cardinality of $L_k$, i.e., the number of one in $E_{1k},\ldots,E_{k-1,k}$. It follows from stationarity and symmetry, together with the fact that the volume measure is translationally invariant in a torus, that $E_{1k},\ldots,E_{k-1,k}$ are independent and identically distributed Bernoulli random variables with parameter $q$, where $q = \Pr(X \sim Y)$ for any two points $X$ and $Y$ distributed on $\mathbb{T}^d$. As a consequence, $L$ has a binomial distribution with parameters $k-1$ and $q$.

Define $\mathcal{I}(X^{\ell}) \triangleq \bigcap_{i=1}^{\ell} \mathcal{B}_{X_i}(r)$ for $\ell \geq 1$ and $\mathcal{I}(X^{0})\triangleq \mathbb{T}^d$. The variable pairs $(X_i, E_{ik})$, $1 \leq i \leq k-1$, are exchangeable. Hence, the conditional expectation of the volume of $\mathcal{I}(X^{k-1},L_k)$ with respect to $L_k$ depends only on the cardinality of $L_k$, which is $L$:
\[\ex \left[\vol \left( \mathcal{I}(X^{k-1},L_k) \right)\mid L_k\right]=\ex \left[\vol \left( \mathcal{I}(X^{L}) \right)\mid L\right].\]
Now, by taking the expectation on both sides, we obtain
\begin{align*}
    \ex \left[\vol \left( \mathcal{I}(X^{k-1},L_k) \right)\right]
     &=\ex\left[ \ex \left[\vol \left( \mathcal{I}(X^{k-1},L_k) \right)\mid L_k\right] \right] \nonumber \\
        &=\ex\left[\ex \left[\vol \left( \mathcal{I}(X^{L}) \right)\mid L \right]\right] \nonumber\\
        &=\ex[v(L)], \label{eq:v1}
\end{align*}
where $v(\ell) = \ex \left[\vol \left( \mathcal{I}(X^{L}) \right)\mid L=\ell\right]$ is the average volume of the intersection region formed by $\ell$ balls of radius $r$ conditioned on the event that $X_k$ lies in that region. By definition, $v(\ell)=1$ for $\ell=0$.

Conditioned on the event $L= \ell$, where $1 \leq \ell \leq k-1$ and the location $X_k$, the random variables $X^L=(X_1, \ldots, X_L)$ are independently and uniformly distributed on the ball $\mathcal{B}_{X_k}(r)$ centered at $X_k$. Because of the translation invariance of the volume measure on torus, the average volume $v(\ell)$ is nothing but the average volume of $\mathcal{I}(X^{L})$ when the random variables in $X^L$ are independently and uniformly distributed on a ball of radius $r$ around a fixed point of $\mathbb{T}^d$:
\begin{align*}
    v(\ell) &= \ex \left[\vol \left( \mathcal{I}(X^{L}) \right)\mid L=\ell\right]\nonumber \\
    &=\ex\left[\ex \left[\vol \left( \mathcal{I}(X^{L}) \right)\mid L=\ell, X_k\right] \right]\\
    & = \ex \left[\vol \big( \mathcal{I}(\tilde{X}^{\ell}) \big)\right],
\end{align*}
where $\tilde{X}^{\ell} = (\tilde{X}_1, \ldots, \tilde{X}_{\ell})$ are independently and uniformly distributed points in $\mathcal{B}_{c}(r)$ with $c=(1/2, \ldots, 1/2) \in \mathbb{T}^d$.

The asymptotic behavior of the average volume $v(\ell)$ follows from a result of Richey and Sarkar~\cite[Proposition~1.1]{Richey_Sarkar_2022} that says that for a  Poisson point process $\tilde{X}^{\lambda}$ with intensity parameter $\lambda$ in a ball of radius $r$ (say, centered at $c$) in $\mathbb{R}^d$,
\begin{align}
     \ex \left[\vol \big( \mathcal{I}^{\lambda} \big)\right] \sim C_{d,r}\cdot \lambda^{-d},
\end{align}
where $C_{d,r}$ is a constant that depends only on the dimension $d$ and radius $r$, and $\mathcal{I}^{\lambda}=\mathcal{B}_c(r) \bigcap\bigg[\bigcap \limits_{X \in \tilde{X}^{\lambda}} \mathcal{B}_{X}(r)\bigg]$. The authors remark that the same is the scaling limit even in the fixed count version with $\ell$ points, i.e., $\ex \left[\vol \big( \mathcal{I}(\tilde{X}^{\ell}) \big)\right] \sim C \cdot \ell^{-d}$ for a constant $C$. In fact, it can be verified by applying the standard de-poissonization argument for the expectation of monotonically decreasing functions\footnote{ It is evident that $\ex \left[\vol \big( \mathcal{I}(\tilde{X}^{\ell}) \big)\right]$  is a non-increasing function of the integer $l$, i.e., $\ex \left[\vol \big( \mathcal{I}(\tilde{X}^{\ell+1}) \big)\right] \leq \ex \left[\vol \big( \mathcal{I}(\tilde{X}^{\ell}) \big)\right]$ for $\ell \geq 1$ , as the average volume of the intersection decreases when we add a new random point $X_{\ell+1}$ into ${X_1, \ldots, X_{\ell}}$. 
}\cite[Theorem 5.10]{Mitzenmacher_Upfal_2005} that $\ex \left[\vol \big( \mathcal{I}(\tilde{X}^{\ell}) \big)\right] \leq 2\cdot\ex \left[\vol \big( \mathcal{I}^{\ell} \big)\right]$, where  we choose the intensity parameter to be $\ell$. Therefore, we have $\ex \left[\vol \big( \mathcal{I}(\tilde{X}^{\ell}) \big)\right] = O( \ell^{-d})$. When $r \leq \frac{1}{4}$, the volume of the intersection $\mathcal{I}(\tilde{X}^{\ell})$ in a torus is the same\footnote{When $r \leq \frac{1}{4}$ and all the balls contain $c=(1/2, \ldots, 1/2)$, they are strictly contained in $[0,1]^d$ without wrapped around the boundaries. So, the region covered by the intersection is exactly the same as that in $\mathbb{R}^d$.} as that in $\mathbb{R}^d$. Therefore, we have
\begin{align}
v(\ell) = O( \ell^{-d}). \label{eq:fixed volume}
\end{align}

\begin{lemma}\label{lem:bin}
    Let $v(\ell)$ denote the volume of an intersection of $\ell$ balls independently distributed on $\mathcal{B}_{c}(r)$ in a torus, where $c=(1/2, \ldots, 1/2)$ and $r \leq 1/4$, and let $L\sim\mathsf{Bin}(k-1,p)$.  Then,
    \[
        \ex[v(L)] = O(k^{-d}).
    \]
\end{lemma}
\begin{proof}
The result follows by considering the high probability event $\{L \geq (k-1)\frac{p}{2}$\} and expanding the expectation in terms of it:
\begin{align}
        \ex[v(L)] &= \ex\left[v(L)\mathbbm{1}\left\{L < (k-1)\frac{p}{2}\right\}\right]\nonumber\\
        & \mkern 140mu+ \ex\left[v(L)\mathbbm{1}\left\{L \geq (k-1)\frac{p}{2}\right\}\right]\nonumber\\
        &\leq \pr\left(L \leq (k-1)\frac{p}{2}\right) + C. \left((k-1)\frac{p}{2}\right)^{-d} \label{eq:l1}\\
        &\leq e^{-(k-1)\frac{p^2}{2}} + C. \left((k-1)\frac{p}{2}\right)^{-d}\label{eq:l2}\\
        & \leq  \tilde{C}. k^{-d}\label{eq:l3},
    \end{align}
where the first term of \eqref{eq:l1} is due to the fact that $v(L) \leq 1$ and the second term follows from combining \eqref{eq:fixed volume}, which is $v(l) \leq C \cdot l^{-d}$ for some constant $C>0$,  with $l \geq (k-1)\frac{p}{2}$; \eqref{eq:l2} applies the Hoeffding's inequality; and in \eqref{eq:l3}, $\tilde{C}$ is a large enough constant.
\end{proof}
 By using Lemma~\ref{lem:bin} in \eqref{eq:v1} and combining it with \eqref{eq:ent2} and \eqref{eq: ent1}, we obtain
\begin{align*}
H(G_m)&\geq - \sum_{k=1}^{m} \log \left( \tilde{C}. k^{-d}\right) \\
&= d \sum_{k=1}^{m} \log k - m \log \tilde{C}\\
& = dm\log m - o(m \log m),
\end{align*}
completing the proof of the theorem.

\section{Conclusion} \label{sec:conc}
In this work, we studied the asymptotic behavior of the entropy of a hard random geometric graph on the $d$-dimensional unit cube and unit torus. For all fixed connection range values, we derived upper bounds on $H(G_m)$. In a few cases, we proved that the entropy asymptotically behaves like $dm\log m$, which was then used to deduce the structural entropy of an RGG. 

It should be noted that the proof technique of our complete result on $\mathbb{T}^d$ can be extended to the other connection range values and the unit cube. The reason for restricting to a special case is to make the technical analysis simpler. We strongly believe and conjecture that the leading order terms in the upper bounds of Thm.~\ref{thm:ent_rgg_upper_euclidean} and Thm.~\ref{thm:ent_rgg_upper_torus} give the right behaviour for the entropy asymptotically.

Currently we considered only a connection range value $r$ that is fixed; however, it is an interesting open question how the entropy scales when $r$ depends on the number of vertices $m$. Also, from a practical standpoint, it is of interest to design a compression scheme that attains a compression length close the entropy $dm \log m$, which we leave for future work.

\bibliographystyle{IEEEtran}
\bibliography{IEEEabrv,References}

\begin{thebibliography}{10}
\providecommand{\url}[1]{#1}
\csname url@samestyle\endcsname
\providecommand{\newblock}{\relax}
\providecommand{\bibinfo}[2]{#2}
\providecommand{\BIBentrySTDinterwordspacing}{\spaceskip=0pt\relax}
\providecommand{\BIBentryALTinterwordstretchfactor}{4}
\providecommand{\BIBentryALTinterwordspacing}{\spaceskip=\fontdimen2\font plus
\BIBentryALTinterwordstretchfactor\fontdimen3\font minus \fontdimen4\font\relax}
\providecommand{\BIBforeignlanguage}[2]{{%
\expandafter\ifx\csname l@#1\endcsname\relax
\typeout{** WARNING: IEEEtran.bst: No hyphenation pattern has been}%
\typeout{** loaded for the language `#1'. Using the pattern for}%
\typeout{** the default language instead.}%
\else
\language=\csname l@#1\endcsname
\fi
#2}}
\providecommand{\BIBdecl}{\relax}
\BIBdecl

\bibitem{PenroseBook}
M.~Penrose, \emph{Random Geometric Graphs}.\hskip 1em plus 0.5em minus 0.4em\relax Oxford: Oxford University Press, 2003.

\bibitem{haenggi_2012}
M.~Haenggi, \emph{Stochastic Geometry for Wireless Networks}.\hskip 1em plus 0.5em minus 0.4em\relax Cambridge University Press, 2012.

\bibitem{dani2024}
\BIBentryALTinterwordspacing
V.~Dani, J.~Díaz, T.~P. Hayes, and C.~Moore, ``Reconstruction of random geometric graphs: Breaking the {$\Omega(r)$} distortion barrier,'' \emph{European Journal of Combinatorics}, vol. 121, p. 103842, 2024, sI:EuroComb 2021. [Online]. Available: \url{https://www.sciencedirect.com/science/article/pii/S0195669823001609}
\BIBentrySTDinterwordspacing

\bibitem{mcdiarmid_diskgraph_2014}
\BIBentryALTinterwordspacing
C.~McDiarmid and T.~Müller, ``The number of disk graphs,'' \emph{European Journal of Combinatorics}, vol.~35, pp. 413--431, 2014. [Online]. Available: \url{https://www.sciencedirect.com/science/article/pii/S0195669813001479}
\BIBentrySTDinterwordspacing

\bibitem{choi_structural_entropy}
Y.~Choi and W.~Szpankowski, ``Compression of graphical structures: Fundamental limits, algorithms, and experiments,'' \emph{IEEE Trans.\ Inf.\ Theory}, vol.~58, no.~2, pp. 620--638, Feb. 2012.

\bibitem{delgosha_universal_compression}
P.~Delgosha and V.~Anantharam, ``Universal lossless compression of graphical data,'' \emph{IEEE Trans.\ Inf.\ Theory}, vol.~66, no.~11, pp. 6962--6976, Nov. 2020.

\bibitem{nikpey_graph_side}
H.~Nikpey, S.~Sarkar, and S.~S. Bidokhti, ``Compression with unlabeled graph side information,'' in \emph{Proc.\ IEEE Int.\ Symp.\ Inf.\ Theory (ISIT)}, Taipei, Taiwan, Jun. 2023, pp. 713--718.

\bibitem{martin_rate_distortion_sbm}
M.~W. Wafula, P.~K. Vippathalla, J.~Coon, and M.-A. Badiu, ``Rate-distortion function of the stochastic block model,'' \emph{arXiv preprint arXiv:2309.14464}, 2023.

\bibitem{bustin_lossy}
R.~Bustin and O.~Shayevitz, ``On lossy compression of directed graphs,'' \emph{IEEE Trans.\ Inf.\ Theory}, vol.~68, no.~4, pp. 2101--2122, Apr. 2022.

\bibitem{mihai2021structural}
M.-A. Badiu and J.~P. Coon, ``Structural complexity of one-dimensional random geometric graphs,'' \emph{IEEE Trans.\ Inf.\ Theory}, vol.~69, no.~2, pp. 794--812, Sep. 2023.

\bibitem{Abbe2016GraphClusters}
E.~Abbe, ``{Graph compression: The effect of clusters},'' in \emph{Proc.\ 54th Annu. Allerton Conf.\ Commun.\ Contr.\ Comput.}, Monticello, IL, USA, Sep. 2016, pp. 1--8.

\bibitem{coon2018entropy}
J.~P. Coon, C.~P. Dettmann, and O.~Georgiou, ``Entropy of spatial network ensembles,'' \emph{Physical Review E}, vol.~97, no.~4, p. 042319, 2018.

\bibitem{mihai_2018_dist_rgg}
M.-A. Badiu and J.~P. Coon, ``On the distribution of random geometric graphs,'' in \emph{Proc.\ IEEE Int.\ Symp.\ Inf.\ Theory (ISIT)}, Vail, CO, USA, Jun. 2018, pp. 2137--2141.

\bibitem{praneeth_side_2025}
P.~K. Vippathalla, M.-A. Badiu, and J.~P. Coon, ``Graph compression with side information at the decoder,'' in \emph{Proc.\ IEEE Int.\ Symp.\ Inf.\ Theory (ISIT)}, Ann Arbor, MI, USA, Jun. 2025, pp. 1--6.

\bibitem{lossy_spatial_24}
P.~K. Vippathalla, M.~W. Wafula, M.-A. Badiu, and J.~P. Coon, ``On the lossy compression of spatial networks,'' in \emph{Proc.\ IEEE Int.\ Symp.\ Inf.\ Theory (ISIT)}, Athens, Greece, Jul. 2024, pp. 416--421.

\bibitem{cover_thomas}
T.~M. Cover and J.~A. Thomas, \emph{Elements of Information Theory}, 2nd~ed.\hskip 1em plus 0.5em minus 0.4em\relax USA: Wiley-Interscience, 2006.

\bibitem{alon_1994}
N.~Alon and A.~Orlitsky, ``A lower bound on the expected length of one-to-one codes,'' \emph{IEEE Trans.\ Inf.\ Theory}, vol.~40, no.~5, pp. 1670--1672, Sep. 1994.

\bibitem{Paton2022labelled_entropy}
J.~Paton, H.~Hartle, H.~Stepanyants, P.~van~der Hoorn, and D.~Krioukov, ``Entropy of labeled versus unlabeled networks,'' \emph{Phys.\ Rev.\ E}, vol. 106, 054308, Nov. 2022.

\bibitem{Richey_Sarkar_2022}
J.~Richey and A.~Sarkar, ``Intersections of random sets,'' \emph{Journal of Applied Probability}, vol.~59, no.~1, p. 131–151, 2022.

\bibitem{luczak_pag_asym}
T.~Łuczak, A.~Magner, and W.~Szpankowski, ``Asymmetry and structural information in preferential attachment graphs,'' \emph{Random Struct.\ Alg.}, vol.~55, no.~3, pp. 696--718, Mar. 2019.

\bibitem{warren1968}
H.~E. Warren, ``Lower bounds for approximation by nonlinear manifolds,'' \emph{Trans.\ Am.\ Math.\ Soc.}, vol. 133, no.~1, pp. 167--178, Aug. 1968.

\bibitem{Mitzenmacher_Upfal_2005}
M.~Mitzenmacher and E.~Upfal, \emph{Probability and Computing: Randomized Algorithms and Probabilistic Analysis}.\hskip 1em plus 0.5em minus 0.4em\relax Cambridge University Press, 2005.

\bibitem{orderstatlecturenotes}
O.~Hryniv, ``Lecture notes on order statistics,'' \url{https://www.maths.dur.ac.uk/stats/courses/Proba34/1617Probability34H.pdf}, 2017.

\bibitem{Holst_1980}
L.~Holst, ``On the lengths of the pieces of a stick broken at random,'' \emph{J. Appl. Probab.}, vol.~17, no.~3, p. 623–634, 1980.

\end{thebibliography}

\appendices
\section{A general approach for proving a lower bound on the entropy of an RGG}\label{app:gen approach}
For a lower bound, we can start by expanding the entropy. Since
$
    H(G_n)  = H\left(G_{n-1}, E_{1,n}, E_{2,n}, \ldots, E_{n-1,n}\right)
    = H\left(G_{n-1}\right)+H\left(E_{1,n}, E_{2,n}, \ldots, E_{n-1,n}|G_{n-1}\right),
$ for all $2 \leq n \leq m$, we have
\begin{align}
    &H(G_n) - H(G_{n-1}) \nonumber\\ &\geq H\left(E_{1,n}, E_{2,n}, \ldots, E_{n-1,n}|G_{n-1} , X_1, X_2, \ldots, X_{n-1}\right)\nonumber\\
    &= H\left(E_{1,n}, E_{2,n}, \ldots, E_{n-1,n}|X_1, X_2, \ldots, X_{n-1}\right)\label{ineq:X defines G 2}\\
    &= H\left(E_{1,0}, E_{2,0}, \ldots, E_{n-1,0}|X_1, X_2, \ldots, X_{n-1}\right)\label{ineq:relabel 2}, 
\end{align}
where \eqref{ineq:X defines G 2} is due to the fact that the locations $X_1, \ldots, X_{n-1}$ completely determine the graph $G_{n-1}$, and in the equality \eqref{ineq:relabel 2}, relabelling the node $n$ with the random location $X_n$ by $0$ with the random location $X_0$ does not change the value of the conditional entropy. Since $H(G_1)=0$, in essence, we have the lower bound 
\begin{align}
 H(G_m) & = \sum_{n=2}^{m} \left[H(G_n) -H(G_{n-1})\right] \nonumber \\ &\geq \sum_{n=2}^{m} H\left(E_{1, 0},\ldots,E_{n-1, 0}\big| X_1, \ldots,X_{n-1}\right)\nonumber \\ & = \sum_{n=1}^{m-1} H\left(E_{1, 0},\ldots,E_{n, 0}\big| X_1, \ldots,X_{n}\right)\label{eq:graph_ent:recursive}.
\end{align}
The asymptotic behaviour of $H(G_m)$ is governed by the behaviour of $H(E_{1 0},\ldots,E_{m 0}\big| X_1, \allowbreak  \ldots,X_{m})$ as $m \to \infty$. It is enough to proof a lower bound of the form $ c_r \cdot d\log m - o(\log m)$ with some constant $c_r$ for all large enough $m$ on the above conditional entropy, i.e., $H\left(E_{1 0},\ldots,E_{m 0}\big| X_1, \ldots,X_{m}\right) \geq  c_r \cdot d\log m - o(\log m)$ for all $m \geq M$, where $M$ is some large integer, because this immediately yields 
\begin{align}
    H(G_m) &\geq \sum_{n=2}^{m-1} H\left(E_{1 0},\ldots,E_{n 0}\big| X_1, \ldots,X_{n}\right)  \nonumber\\ &\geq   \sum_{n=M}^{m-1}[ c_r \cdot d\log n - o(\log n)]\nonumber\\
    & \geq  c_r \cdot d\int_{M-1}^{m-1} \log x \, dx - o(m \log m)\nonumber\\ 
    &=  c_r \cdot d m \log m - o(m \log m)\label{eq:entropy low arg}.
\end{align}

Let us denote the domain $[0,1]^d$ of the $d$-dimensional unit cube and unit torus by $V_d$. For a fixed edge configuration $(e_{10},\ldots,e_{m0})\in \{0,1\}^m$ and a fixed point configuration $(x_1,\ldots,x_m) \in V_d^m$, we have
\begin{align}
    &P\left(E_{1 0}=e_{10},\ldots,E_{m0}=e_{m0}\big| X_1=x_1, \ldots,X_{m}=x_m\right)\nonumber\\
    &= \int_{V_d}P\left(E_{1 0}=e_{10},\ldots,E_{m0}=e_{m0}\big| X_1=x_1,\right. \nonumber \\ & \hspace{12em} \left. \ldots,X_{m}=x_m, X_0=x_0\right)dx_0\nonumber\\
    &= \int_{V_d}\mathbbm{1}_{A_{e_{10},\ldots,e_{m0},x_1,\ldots,x_m}}(x_0)dx_0\nonumber\\ &=\vol(A_{e_{10},\ldots,e_{m0},x_1,\ldots,x_m}),\label{eq:prob_volume}
\end{align}
where $A_{e_{10},\ldots,e_{m0},x_1,\ldots,x_m}$ is the set of all points $x_0 \in [0,1]^d$ satisfying the edge configuration $(e_{01},\ldots,e_{0m})$ with the given point configuration $(x_1,\ldots,x_m)$ and $\vol(A_{e_{10},\ldots,e_{m0},x_1,\ldots,x_m})$ is its volume. If $A_{e_{10},\ldots,e_{m0},x_1,\ldots,x_m}$ is an empty set, then  $\vol(A_{e_{10},\ldots,e_{m0},x_1,\ldots,x_m})=0$. As a result, we have
\begin{align}
&H\left(E_{1 0},\ldots,E_{m 0}\big| X_1=x_1, \ldots,X_{m}=x_m\right)\nonumber\\
&=\sum_{(e_{10},\ldots,e_{m0})\in \{0,1\}^m}\vol(A_{e_{10},\ldots,e_{m0},x_1,\ldots,x_m}) \nonumber \\ & \hspace{10em}\times \log\frac{1}{\vol(A_{e_{10},\ldots,e_{m0},x_1,\ldots,x_m})}\nonumber.\\
&= \left.\mathbb{E}_{E_{10},\ldots,E_{m0}}\left[\log\frac{1}{\vol(A_{E_{10},\ldots,E_{m0},X_1,\ldots,X_m})} \right\vert X_1=x_1, \right. \nonumber \\ & \hspace{15em}  \ldots, X_m=x_m  \bigg].
\end{align}
We can then rewrite the required conditional entropy term as follows:
\begin{align}
    &H\left(E_{1 0},\ldots,E_{m 0}\big| X_1, \ldots,X_{m}\right)\nonumber\\
    &= \mathbb{E}_{X_1, \ldots, X_m}\bigg[\mathbb{E}_{E_{10},\ldots,E_{m0}}\bigg[\log\frac{1}{\vol(A_{E_{10},\ldots,E_{m0},X_1,\ldots,X_m})}   \nonumber \\ & \hspace{18em} \bigg\vert X_1, \ldots, X_m  \bigg] \bigg]\nonumber\\
    &=\mathbb{E}\left[\log\frac{1}{\vol(A_{E_{10},\ldots,E_{m0},X_1,\ldots,X_m})} \right]. \label{eq:avg_log_vol}   
\end{align}
In the subsequent proofs, we will bound the expression in \eqref{eq:avg_log_vol}.
\section{Proof of Theorem~\ref{thm:ent_rgg_lower_one_d}: Lower bound on the entropy of the one dimensional RGG $(d=1)$ on $[0,1]$}\label{proof:thm:ent_rgg_lower_one_d}
We will prove this theorem by lower bounding \eqref{eq:avg_log_vol} using the order statistics \cite{orderstatlecturenotes} of the random points $(X_1, \ldots, X_m)$ in the unit interval $[0,1]$. To this end, consider the following lemma.

\begin{lemma}\label{lem:volume_diff_E}
For $0< r <1$, the volume of the set $A_{E_{10},\ldots,E_{m0},X_1,\ldots,X_m}$ for fixed node locations $(X_1, \ldots, X_m)$ is given as follows:
\begin{enumerate}
    \item If $E_{10}=\ldots=E_{m0}=0$, then 
    \begin{align}
&\mkern -24mu \vol(A_{E_{10},\ldots,E_{m0},X_1,\ldots,X_m}) \nonumber\\& \mkern -24mu = [X^{(1)}-X^{(0)}-r]^{+}  + \sum_{i=1}^{m-1}[X^{(i+1)}-X^{(i)}-2r]^{+}\nonumber\\& \hspace{8em}+ [X^{(m+1)}-X^{(m)}-r]^{+};
\end{align}
\item 
If $E_{10}=\ldots=E_{m0}=1$, then 
    \begin{align}
&\mkern -24mu\vol(A_{E_{10},\ldots,E_{m0},X_1,\ldots,X_m}) \nonumber\\&\mkern -24mu= [2r - (\max\{X^{(m)},r\} - \min\{X^{(1)},1-r\})]^{+};
\end{align}
\item  If $E_{10}, \ldots, E_{m0}$ is any configuration with at least one $0$ and one $1$, then 
\begin{align}
\vol(A_{E_{10},\ldots,E_{m0},X_1,\ldots,X_m}) \leq \max_{0\leq i \leq m}(X^{(i+1)}-X^{(i)}),
\end{align} 
\end{enumerate}
where $[a]^{+}=\max\{a,0\}$, $X^{(0)}:=0$, $X^{(m+1)}:=1$, and $(X^{(1)}, X^{(2)}, \ldots, X^{(m)})$ are the ordered points of $(X_1, X_2, \ldots, X_m)$.
\end{lemma}
\begin{proof} 
For ease of notation,  we use $A$ instead of $A_{E_{10},\ldots,E_{m0},X_1,\ldots,X_m}$ in the rest of the proof of this lemma.
 Let us start by proving the first statement. Any $x_0 \in A \cap (X^{(i+1)}, X^{(i)}]$, $1 \leq i \leq m-1$ iff it satisfies the condition that  $ X^{(i)} + r < x_0 < X^{(i+1)} - r$ because all edge random variables take the value $0$. This implies that $\vol(A \cap (X^{(i+1)}, X^{(i)}])= [X^{(i+1)}-X^{(i)}-2r]^{+}$ for $1 \leq i \leq m-1$. Similarly, any $x_0 \in A \cap (X^{(0)}, X^{(1)}]$ (resp. $A \cap (X^{(m)}, X^{(m+1)}]$) satisfies the condition  $ 0= X^{(0)}  < x_0 < X^{(1)} - r$ (resp. $ X^{(m)} + r  < x_0 \leq X^{(m+1)} = 1$). In these cases, we have $\vol(A \cap (X^{(0)}, X^{(1)}])= [X^{(1)}-X^{(0)}-r]^{+}$ and $\vol(A \cap (X^{(m)}, X^{(m+1)}])= [X^{(m+1)}-X^{(m)}-r]^{+}$. As the intervals $(X^{(i+1)}, X^{(i)}]$ are disjoint and cover $(0,1]$, we have
 $$\vol(A)= \sum_{i=0}^m \vol(A \cap (X^{(i+1)}, X^{(i)}] ),$$ 
 proving the statement 1).

 For 2), where $E_{10}=\ldots=E_{m0}=1$, the set is completely determined by the extreme points  $X^{(1)}$ and $X^{(m)}$. Any $x_0 \in A$ iff $x_0 \leq X^{(1)}+r$ and $X^{(m)}-r \leq x_0$. In this case, we have 
 $$\vol(A)= [\min\{X^{(1)}+r, 1\} - \max\{X^{(m)}-r, 0\}]^{+},$$
 whence the second statement follows by bringing out $r$ from $\max$ and $\min$.

 In the third statement, there is at least one $0$ and one $1$ in the edge configuration. It means that there exist  neighbours $X^{(j)}$ and $X^{(j+1)}$ for some $j \in \{1, \dots, m\}$ such that one of the corresponding edge random variables is $0$ and the other is $1$. Let us assume that the edge random variable corresponding to $X^{(j)}$ is $0$ with the other being $1$. In this case, if $x_0 \in A$ then $X^{(j)}+r \leq x_0 \leq X^{(j+1)}+r$. On the other hand, when  the edge random variable corresponding to $X^{(j)}$ is $1$ with the other being $0$,  $x_0 \in A$ implies that $X^{(j)}-r \leq x_0 \leq X^{(j+1)}-r$. In both these case, we have 
 $$\vol(A) \leq X^{(j+1)}-X^{(j)} \leq \max_{0\leq i \leq m}(X^{(i+1)}-X^{(i)}).$$ 
\end{proof}

Now we will consider two separate cases, depending on whether the connection radius $r$ is less than or greater than one half. 

\subsection{$0< r\leq \frac{1}{2}$}
In this case, we can apply Jensen's inequality to $\log\frac{1}{z}$ in  the conditional entropy equation \eqref{eq:avg_log_vol}:
\begin{align}
H\left(E_{1 0},\ldots,\right.& \left. E_{m 0}\big| X_1, \ldots,X_{m}\right)\nonumber\\
&= \mathbb{E}\left[\log\frac{1}{\vol(A_{E_{10},\ldots,E_{m0},X_1,\ldots,X_m})} \right] \nonumber\\
    &\geq \log\frac{1}{\mathbb{E}\left[\vol(A_{E_{10},\ldots,E_{m0},X_1,\ldots,X_m}) \right]}, \label{eq: volume lower bound}
\end{align}
The entropy can be further bounded by deriving an upper bound on the average volume:
\begin{align}
&\mkern -16mu \mathbb{E}\left[\vol(A_{E_{10},\ldots,E_{m0},X_1,\ldots,X_m})\right]\nonumber\\ 
&\mkern -16mu =\mathbb{E}\big[\vol(A_{E_{10},\ldots,E_{m0},X_1,\ldots,X_m})(\mathbbm{1}\{E_{10}=\ldots=E_{m0}=0\} \nonumber\\ &\hspace{10em} + \mathbbm{1}\{E_{10}=\ldots=E_{m0}=1\})\big]\nonumber\\
& + \mathbb{E}\big[\vol(A_{E_{10},\ldots,E_{m0},X_1,\ldots,X_m})\nonumber\\ &\hspace{7em} \times\mathbbm{1}\{\exists \ i\neq j  \text{ s. t. }E_{i0}=0, E_{j0}=1\} \big]\nonumber\\
    &\mkern -16mu \leq \mathbb{P}(E_{10}=\ldots=E_{m0}=0) + \mathbb{P}(E_{10}=\ldots=E_{m0}=1)  \nonumber\\    &\hspace{8em} +  \mathbb{E}\left[\max_{1\leq i \leq m-1}(X^{(i+1)}-X^{(i)})\right], \label{eq:sum_three_terms}
\end{align}
where the last inequality uses the fact that the volume is less than one for the first two terms and Lemma~\ref{lem:volume_diff_E} along with the inequality $\mathbbm{1} \leq 1$ (as the factor $\vol \geq 0$) for the third term.\\

Let us consider the first term  of \eqref{eq:sum_three_terms}.
\begin{align}
    &\mathbb{P}(E_{10}=\ldots=E_{m0}=0)\nonumber\\ & = \mathbb{E}_{X_1, \ldots, X_m}\left[\mathbb{P}\left(E_{10}=\ldots=E_{m0}=0 \big| X_1, \ldots,X_{m}\right)\right]\nonumber\\ &
    = \mathbb{E}_{X_1, \ldots, X_m}\left[\vol(A_{0,\ldots,0,X_1,\ldots,X_m})\right]\label{eq:prob_zero_1}\\
    & = \mathbb{E}\bigg[[X^{(1)}-X^{(0)}-r]^{+} + \sum_{i=1}^{m-1}[X^{(i+1)}-X^{(i)}-2r]^{+}\nonumber\\ & \hspace{10em}  + [X^{(m+1)}-X^{(m)}-r]^{+}\bigg]\label{eq:prob_zero_2}\\
    & \leq \sum_{i=0}^{m} \mathbb{E}\left[[X^{(i+1)}-X^{(i)}-r]^{+}\right]\nonumber\\
    &= m\cdot \mathbb{E}\left[[X^{(1)}-X^{(0)}-r]^{+}\right]\label{eq:prob_zero_3}\\
    &=  m\cdot \mathbb{E}\left[[X^{(1)}-X^{(0)}-r]\mathbbm{1}\{X^{(1)}-X^{(0)}\geq r\}\right] \nonumber \\
    & \leq m (1-r) \cdot \mathbb{E}\left[\mathbbm{1}\{X^{(1)}-X^{(0)}\geq r\}\right]\nonumber\\
    & = m (1-r) \cdot \mathbb{P}\left(X^{(1)}-X^{(0)}\geq r\right)\nonumber\\
    &= m(1-r)^{m+1}, \label{eq:prob_es_zero}
\end{align}
where \eqref{eq:prob_zero_1} follows from \eqref{eq:prob_volume}; \eqref{eq:prob_zero_2} applies Lemma~\eqref{lem:volume_diff_E}; and the equalities \eqref{eq:prob_zero_3} and \eqref{eq:prob_es_zero} are due to the fact the gaps $X^{(i+1)}-X^{(i)}$ are identically distributed with $\mathbb{P}\left(X^{(i+1)}-X^{(i)}\geq r\right) = (1-r)^m$.\\

We can bound the second term of \eqref{eq:sum_three_terms} in a similar way.
\begin{align}
    &\mathbb{P}(E_{10}=\ldots=E_{m0}=1)\nonumber\\
    & = \mathbb{E}_{X_1, \ldots, X_m}\left[\mathbb{P}\left(E_{10}=\ldots=E_{m0}=1 \big| X_1, \ldots,X_{m}\right)\right]\nonumber\\ 
    &
    = \mathbb{E}_{X_1, \ldots, X_m}\left[\vol(A_{1,\ldots,1,X_1,\ldots,X_m})\right]\nonumber\\ 
    &
    \leq  \mathbb{E}\Big[\vol(A_{1,\ldots,1,X_1,\ldots,X_m})  \Big( \mathbbm{1}\{X^{(1)} < 1-r, X^{(m)} \geq r \} \nonumber\\ 
    & \hspace{7 em}+ \mathbbm{1}\{X^{(1)} \geq 1-r \}  + \mathbbm{1}\{X^{(m)} < r\}\Big)\Big]\nonumber\\
    &\leq  \mathbb{E}\Big[\vol(A_{1,\ldots,1,X_1,\ldots,X_m}) \mathbbm{1}\{X^{(1)} < 1-r, X^{(m)} \geq r \}\Big] 
    \nonumber\\ 
    & \hspace{7 em} + \mathbb{P}(X^{(1)} \geq 1-r) + \mathbb{P}(X^{(m)} < r)\nonumber\\
    & = \mathbb{E}\Big[\big[2r - (X^{(m)} - X^{(1)})\big]^{+}\mathbbm{1}\{X^{(1)} < 1-r, X^{(m)} \geq r \}\Big] \nonumber\\ 
    & \hspace{7 em} + \mathbb{P}(X^{(1)} \geq 1-r)  +  \mathbb{P}(X^{(m)} < r)\label{eq:prob_one_2}\\
    & = \mathbb{E}\Big[\big[2r - (X^{(m)} - X^{(1)})\big] \mathbbm{1}\{2r \geq X^{(m)} - X^{(1)}  \}\nonumber\\ & \mkern 180 mu \times \mathbbm{1}\{X^{(1)} < 1-r, X^{(m)} \geq r \}\Big]  \nonumber \\ &  \hspace{7 em} + \mathbb{P}(X^{(1)} \geq 1-r) + \mathbb{P}(X^{(m)} < r),\label{eq:prob_one_f}
\end{align}
where \eqref{eq:prob_one_2} uses Lemma~\ref{lem:volume_diff_E}. If the connection radius $r \leq \frac{1}{2}$, then we can bound \eqref{eq:prob_one_f} as follows.
\begin{align}
&\mathbb{P}(E_{10}=\ldots=E_{m0}=1)\nonumber\\
    &\leq \mathbb{E}\Big[\mathbbm{1}\{2r \geq X^{(m)} - X^{(1)}  \} \mathbbm{1}\{X^{(1)} < 1-r, X^{(m)} \geq r \}\Big]  \nonumber \\ &  \hspace{8 em} + \mathbb{P}(X^{(1)} \geq 1-r) + \mathbb{P}(X^{(m)} < r)\nonumber\\
    &\leq \mathbb{E}\left[\mathbbm{1}\{2r \geq X^{(m)} - X^{(1)}  \}\right] + \mathbb{P}(X^{(1)} \geq 1-r)\nonumber \\ & \mkern 290 mu  + \mathbb{P}(X^{(m)} < r)\nonumber\\
    & \leq  \mathbb{P}(2r \geq X^{(m)} - X^{(1)})+\mathbb{P}(X^{(1)} \geq 1-r)+ \mathbb{P}(X^{(m)} < r)\nonumber\\
    &= (m(1-2r)+2r)(2r)^{m-1} + r^m + r^m, \label{eq:r less half}
\end{align}
where \eqref{eq:r less half} follows from \cite[(4.10)]{orderstatlecturenotes} and \cite[(4.3)]{orderstatlecturenotes}.\\

The third term \cite[Theorem~2.2]{Holst_1980} in \eqref{eq:sum_three_terms} is  
\begin{align}
    \mathbb{E}\left[\max_{1\leq i \leq m-1}(X^{(i+1)}-X^{(i)})\right] = \frac{1}{m}\sum_{i=0}^{m-1} \frac{1}{m-i} \leq \frac{\log m + 1}{m}. \label{eq:average_max}
\end{align}

By combining all these terms, we obtain that 
\begin{align*}
\mathbb{E}&\left[\vol(A_{E_{10},\ldots,E_{m0},X_1,\ldots,X_m})\right] \nonumber\\ &\leq m (1-r)^{m+1}+ (m(1-2r)+2r)(2r)^{m-1} + 2r^m \nonumber\\ & \mkern 340 mu+ \frac{\log m + 1}{m} \\&\leq  \frac{\log (m+1)}{m},  
\end{align*} 
where the last inequality holds for all large enough $m$. This proves that
\begin{align*}
&H\left(E_{1 0},\ldots,E_{m 0}\big| X_1, \ldots,X_{m}\right) \nonumber\\& =\mathbb{E}\left[\log\frac{1}{\vol(A_{E_{10},\ldots,E_{m0},X_1,\ldots,X_m})} \right] \nonumber\\
    &\geq \log\frac{1}{\mathbb{E}\left[\vol(A_{E_{10},\ldots,E_{m0},X_1,\ldots,X_m}) \right]}\\
    & \geq \log \left(\frac{m}{\log(m+1)}\right)= \log m - \log(\log(m+1))
    \end{align*}
    for all large enough $m$. By using this result in the argument of  \eqref{eq:entropy low arg}, we prove the lower bound
$$H(G_m)\geq m\log m - o(m \log m)$$
for the connection range $0 < r \leq \frac{1}{2}$.

\subsection{$\frac{1}{2} \leq r < 1$}
In contrast to the previous case, it is not useful to directly apply Jensen's inequality as in \eqref{eq: volume lower bound} to deal with the average volume since it does not converge to zero as $m \to \infty$ for $\frac{1}{2} \leq r < 1$, hence yielding a constant lower bound on the conditional entropy rather than one that grows like $\log m$.  In fact, the average volume converges to a non-zero constant. To see this, consider
\begin{align}
&\mathbb{E}\left[\vol(A_{E_{10},\ldots,E_{m0},X_1,\ldots,X_m}) \right]\nonumber\\ &\quad = \mathbb{E}\left[\mathbb{P}(E_{10},\ldots,E_{m0}|X_1, \ldots, X_m)\right] \nonumber \\
& \quad = \mathbb{E}\left[\mathbb{P}(E_{10},\ldots,E_{m0}|X_1, \ldots, X_m)(\mathbbm{1}_{\mathcal{E}_m} + \mathbbm{1}_{\mathcal{E}^c_m})\right],\label{eq:vol radius greater}
\end{align}
where $\mathcal{E}_m \triangleq \{E_{10}=\ldots=E_{m0}=1\}^c$. Observe that when $r \geq \frac{1}{2}$, any random point $X_0$ falling in the interval $[1-r, r]$ will be connected to the rest of the nodes through edges. Therefore, $\mathbb{P}(E_{10}=\ldots=E_{m0}=1|X_1, \ldots, X_m) \geq 2r-1$ and $\mathbb{P}(\mathcal{E}^c_m)=\mathbb{P}(E_{10}=\ldots=E_{m0}=1)= \mathbb{E}\left[\mathbb{P}\left(E_{10}=\ldots=E_{m0}=1 \big| X_1, \ldots,X_{m}\right)\right] \geq 2r-1$ which combined with \eqref{eq:vol radius greater} gives
\begin{align}
&\mathbb{E}\left[\vol(A_{E_{10},\ldots,E_{m0},X_1,\ldots,X_m}) \right]\nonumber\\ 
& \quad = \mathbb{E}\left[\mathbb{P}(E_{10},\ldots,E_{m0}|X_1, \ldots, X_m)(\mathbbm{1}_{\mathcal{E}_m} + \mathbbm{1}_{\mathcal{E}^c_m})\right],
\nonumber\\ 
& \quad \geq \mathbb{E}\left[\mathbb{P}(E_{10},\ldots,E_{m0}|X_1, \ldots, X_m)\mathbbm{1}_{\mathcal{E}^c_m}\right] \nonumber\\ 
& \quad \geq (2r-1)\mathbb{P}(\mathcal{E}_m) \nonumber\\
& \quad \geq (2r-1)^2.\label{eq:vol radius greater 1}
\end{align}

Despite this, we can show that the conditional entropy can be large by focusing on the event  $\mathcal{E}_m$. 
Let us start with the expression of the conditional entropy. Recall from \eqref{eq:avg_log_vol} that
\begin{align}
&H\left(E_{1 0},\ldots,E_{m 0}\big| X_1, \ldots,X_{m}\right) \nonumber\\
&=\mathbb{E}\left[\log\frac{1}{\vol(A_{E_{10},\ldots,E_{m0},X_1,\ldots,X_m})} \right] \nonumber\\
&\geq \mathbb{E}\left[\log\frac{1}{\vol(A_{E_{10},\ldots,E_{m0},X_1,\ldots,X_m})} \mathbbm{1}_{\mathcal{E}_m} \right]\nonumber\\
&\geq \mathbb{P}(\mathcal{E}_m) \log\frac{\mathbb{P}(\mathcal{E}_m)}{\mathbb{E}\left[ \vol(A_{E_{10},\ldots,E_{m0},X_1,\ldots,X_m})\mathbbm{1}_{\mathcal{E}_m} \right]}\label{eq:lower_r_half_1}\\
&= \mathbb{P}(\mathcal{E}_m) \log \mathbb{P}(\mathcal{E}_m) \nonumber\\ &\hspace{5 em} + \mathbb{P}(\mathcal{E}_m) \log\frac{1}{\mathbb{E}\left[ \vol(A_{E_{10},\ldots,E_{m0},X_1,\ldots,X_m})\mathbbm{1}_{\mathcal{E}_m} \right]}\nonumber\\
&\geq \frac{1}{e} \log \frac{1}{e} +  \mathbb{P}(\mathcal{E}_m) \log\frac{1}{\mathbb{E}\left[ \vol(A_{E_{10},\ldots,E_{m0},X_1,\ldots,X_m})\mathbbm{1}_{\mathcal{E}_m} \right]}\nonumber
\end{align}
where the inequality \eqref{eq:lower_r_half_1} follows from Jensen's equality of the form $\mathbb{E}[f(X) \mathbbm{1}_{\mathcal{A}}]= \mathbb{P}(A)\mathbb{E}[f(X) \mid \mathcal{A}] \geq \mathbb{P}(A)f(\mathbb{E}[X \mid \mathcal{A}])= \mathbb{P}(A)f\left(\frac{\mathbb{E}[X \mathbbm{1}_{\mathcal{A}}]}{\mathbb{P}(A)}\right)$ with a convex function $f$ to $f(z)=\log \frac{1}{z}$.

In order to find a lower bound on $\mathbb{P}(\mathcal{E}_m) =1- \mathbb{P}(E_{10}=\ldots=E_{m0}=1)$, we will make use of \eqref{eq:prob_one_f} as the steps in \eqref{eq:r less half} will not work for $r \geq \frac{1}{2}$. 
\begin{align}
    &\mathbb{P}(E_{10}=\ldots=E_{m0}=1)\nonumber\\
     & = \mathbb{E}\Big[\big[2r - (X^{(m)} - X^{(1)})\big] \mathbbm{1}\{2r \geq X^{(m)} - X^{(1)}  \}\nonumber\\ & \mkern 180 mu \times \mathbbm{1}\{X^{(1)} < 1-r, X^{(m)} \geq r \}\Big]  \nonumber \\ &  \hspace{7 em} + \mathbb{P}(X^{(1)} \geq 1-r) + \mathbb{P}(X^{(m)} < r)\nonumber\\
    &\leq  \mathbb{E}_{X_1, \ldots, X_m}\left[2r - (X^{(m)} - X^{(1)})\right] \nonumber\\ 
    & \hspace{6 em} + \mathbb{P}(X^{(1)} \geq 1-r)+ \mathbb{P}(X^{(m)} < r)\nonumber\\
    &= 2r - \mathbb{E}_{X_1, \ldots, X_m}\left[X^{(m)} - X^{(1)}\right] + r^m + r^m\nonumber\\
    &= 2r - \frac{m-1}{m+1} +2 r^m \label{eq:prob_es_are_one}
\end{align}
where \eqref{eq:prob_es_are_one} follows from \cite[(4.10)]{orderstatlecturenotes} and \cite[(4.3)]{orderstatlecturenotes}. Hence, we have
\begin{align*}
    \mathbb{P}(\mathcal{E}_m) &=1- \mathbb{P}(E_{10}=\ldots=E_{m0}=1) \\
    & \geq 1- \left(2r - \frac{m-1}{m+1} +2 r^m\right)\\
    &= 2(1-r)-\left(1 - \frac{m-1}{m+1} +2 r^m\right).
\end{align*}
We also know from \eqref{eq:sum_three_terms}, \eqref{eq:prob_es_zero} and \eqref{eq:average_max} that
\begin{align}
&\mathbb{E}\left[\vol(A_{E_{10},\ldots,E_{m0},X_1,\ldots,X_m})\mathbbm{1}_{\mathcal{E}_m} \right]\nonumber\\   
&=\mathbb{E}\left[\vol(A_{E_{10},\ldots,E_{m0},X_1,\ldots,X_m})(\mathbbm{1}\{E_{10}=\ldots=E_{m0}=0\}) \right]\nonumber\\& \quad +\mathbb{E}\big[\vol(A_{E_{10},\ldots,E_{m0},X_1,\ldots,X_m})\nonumber\\ &\hspace{9 em}\times \mathbbm{1}\{\exists \ i\neq j  \text{ s. t. }E_{i0}=0, E_{j0}=1\} \big]\nonumber\\
    &\leq \mathbb{P}(E_{10}=\ldots=E_{m0}=0) +  \mathbb{E}\left[\max_{1\leq i \leq m-1}(X^{(i+1)}-X^{(i)})\right],\nonumber\\
    &\leq m(1-r)^{m+1} +  \frac{\log m + 1}{m} \nonumber\\
    & \leq \frac{\log (m+1)}{m}, 
\end{align}
where the last inequality holds for all large enough $m$. Thus we have
\begin{align}
 &H\left(E_{1 0},\ldots,E_{m 0}\big| X_1, \ldots,X_{m}\right)\nonumber\\  & \geq \frac{1}{e} \log \frac{1}{e} +  \mathbb{P}(\mathcal{E}_m) \log\frac{1}{\mathbb{E}\left[ \vol(A_{E_{10},\ldots,E_{m0},X_1,\ldots,X_m})\mathbbm{1}_{\mathcal{E}_m} \right]}\nonumber\\
    & \geq \frac{1}{e} \log \frac{1}{e} \nonumber \\
    & \quad + \left[2(1-r)-\left(1 - \frac{m-1}{m+1} +2 r^m\right)\right] \log \left(\frac{m}{\log(m+1)}\right) \nonumber\\
    &= 2(1-r)\log m - o(\log m),
    \end{align}
for all large enough $m$. By using this result in the argument of  \eqref{eq:entropy low arg}, we prove the lower bound
$$H(G_m)\geq 2(1-r) m\log m - o(m \log m)$$
for the connection range $\frac{1}{2} \leq r < 1$,
completing the proof of Theorem~\ref{thm:ent_rgg_lower_one_d}.

\section{An Alternative Proof of Theorem~\ref{thm:ent_rgg_lower_torus} : Lower Bound in the case of $\mathbb{T}^d$} \label{proof:thm:ent_rgg_lower_torus}
We will find an asymptotic expression for $H\left(E_{1 0},\ldots,E_{m 0}\big| X_1, \ldots,X_{m}\right)$ when $r<\frac{1}{4}$, and as described in Appendix~\ref{app:gen approach}, the theorem result will follow. Recall the equation \eqref{eq:avg_log_vol} for the conditional entropy, and by applying Jensen's inequality to the function $\log\frac{1}{z}$, we obtain 
\begin{align}
&H\left(E_{1 0},\ldots,E_{m 0}\big| X_1, \ldots,X_{m}\right)\nonumber \\
& \hspace{6.5em}\geq \log\frac{1}{\mathbb{E}\left[\vol(A_{E_{10},\ldots,E_{m0},X_1,\ldots,X_m}) \right]}.
\end{align}
We will show that the average volume of $A_{E_{10},\ldots,E_{m0},X_1,\ldots,X_m}$ is upper bound by a constant times $\frac{1}{m^d}$ for all large enough $m$.

Recall that  for a fixed point configuration $(x_1,\ldots,x_m)$, the volume of the set of all points $x_0 \in \mathbb{T}^d$ satisfying the edge configuration $(e_{01},\ldots,e_{0m})$ is given by $\vol(A_{E_{10},\ldots,E_{m0},X_1,\ldots,X_m})= \mathbb{P}\left(E_{10},\ldots, E_{m0} \big| X_1, \ldots,X_{m}\right)$. Therefore, We can rewrite the average volume as follows:
\begin{align}
&\mkern -10mu \mathbb{E}\left[\vol(A_{E_{10},\ldots,E_{m0},X_1,\ldots,X_m})\right]\nonumber\\
&=\mathbb{E}\left[\mathbb{P}\left(E_{10},\ldots, E_{m0} \big| X_1, \ldots,X_{m}\right)\right]\nonumber\\
&=\mathbb{E}\left[\mathbb{E}\left[\mathbb{P}\left(E_{10},\ldots, E_{m0} \big| X_1, \ldots,X_{m}\right)| X_1, \ldots,X_{m}\right]\right]\nonumber\\
&=\mathbb{E}\left[\sum_{e_{10}, \ldots, e_{m0}}\mathbb{P}^2\left(e_{10},\ldots, e_{m0} \big| X_1, \ldots,X_{m}\right) \right]\label{eq:p_square}.
\end{align}
The summand in the above expression for a fixed point configuration $(x_1, \ldots, x_m)$ is 
\begin{align}
    \mathbb{P}^2\left(e_{10},\ldots, e_{m0} \big| x_1, \ldots,x_{m}\right)&= \int \limits_x \mathbbm{1}_A(x)\ dx \cdot \int \limits_{x'} \mathbbm{1}_A(x')\  dx'\nonumber\\
    & =\int \limits_x  \int \limits_{x'} \mathbbm{1}_A(x) \cdot\mathbbm{1}_A(x')\ dx'\  dx\nonumber,
\end{align}
where we use the letter $A$ for the set $A_{e_{10},\ldots,e_{m0},x_1,\ldots,x_m}$ to simplify the notation. This is the probability that both the points $x$ and $x'$ that are drawn randomly and independently from $\mathbb{T}^d$ lie in $A$. By substituting this probability in \eqref{eq:p_square}, we obtain
\begin{align}
&\mkern -40mu \mathbb{E}\left[\vol(A_{E_{10},\ldots,E_{m0},X_1,\ldots,X_m})\right]\nonumber \\
&=\mathbb{E}\left[\sum_{e_{10}, \ldots, e_{m0}}\mathbb{P}^2\left(e_{10},\ldots, e_{m0} \big| X_1, \ldots,X_{m}\right) \right]\nonumber\\
&=\mathbb{E}\left[\sum_{e_{10}, \ldots, e_{m0}}\int \limits_x  \int \limits_{x'} \mathbbm{1}_A(x) \cdot\mathbbm{1}_A(x')\ dx'\  dx \right]\nonumber\\
&= \int \limits_x  \int \limits_{x'}  \sum_{e_{10}, \ldots, e_{m0}}\mathbb{E}\left[\mathbbm{1}_A(x) \cdot\mathbbm{1}_A(x')\right]\ dx'\  dx\label{eq: vol_double_indicator}.
\end{align}
Because of the exchangeability of the random variables $(X_1, \ldots, X_m)$, the average $\mathbb{E}\left[\mathbbm{1}_A(x) \cdot\mathbbm{1}_A(x')\right]$ depends only on the number of ones in the edge configuration $(e_{01},\ldots,e_{0m})$ rather than their locations. Therefore, we have 
\begin{align}
    \sum_{e_{10}, \ldots, e_{m0}}\mathbb{E}\left[\mathbbm{1}_A(x) \cdot\mathbbm{1}_A(x')\right]= \sum_{k=0}^m \binom{m}{k} \mathbb{E}\left[\mathbbm{1}_{A_k}(x) \cdot\mathbbm{1}_{A_k}(x')\right],\label{eq:expect_lens}
\end{align}
where $A_k$ is the set $A$ with $e_{10}= \cdots =e_{k0}=1$ and $e_{k+1,0}= \cdots =e_{m0}=0$ for $0 \leq k\leq m$. 

Let ${B}_x \subseteq \mathbb{R}^d$ denote the Euclidean sphere of radius $r$ centered at $x \in \mathbb{R}^d$. Define $L_{x,x'}:=B_x \cap B_{x'}$, which is the lens formed by the intersection in $\mathbbm{R}^d$ of two spheres centered at $x$ and $x'$. Define $C_{x,x'}:=B_x \triangle B_{x'}$, which is the union of two crescents. Similarly, let $\tilde{B}_x \subseteq \mathbb{T}^d$ denote the ball of radius $r$ (with respect to toroidal distance) centered at $x \in \mathbb{T}^d$. Define $\tilde{L}_{x,x'}:=\tilde{B}_x \cap \tilde{B}_{x'}$ and $\widetilde{C}_{x,x'}:=\tilde{B}_x \triangle \tilde{B}_{x'}$ to be the lens and crescents in $\mathbb{T}^d$, respectively.

Observe that $\mathbb{E}\left[\mathbbm{1}_{A_k}(x) \cdot\mathbbm{1}_{A_k}(x')\right]$ is the probability that $X_1,\ldots, X_k \in \tilde{B}_x \cap \tilde{B}_{x'}$ and $X_{k+1},\ldots, X_m \in \tilde{B}_x^c \cap \tilde{B}_{x'}^c$.  With the above notation, \eqref{eq:expect_lens} becomes
\begin{align}
    &\sum_{e_{10}, \ldots, e_{m0}}\mathbb{E}\left[\mathbbm{1}_A(x)\cdot\mathbbm{1}_A(x')\right] = \sum_{k=0}^m \binom{m}{k} \mathbb{E}\left[\mathbbm{1}_{A_k}(x) \cdot\mathbbm{1}_{A_k}(x')\right]\nonumber\\
    &=\sum_{k=0}^m \binom{m}{k} \left[\vol\left(\tilde{L}_{x, x'}\right)\right]^k\left[\vol\left(\tilde{B}^c_x \cap \tilde{B}^c_{x'}\right)\right]^{m-k}\nonumber\\
    &=\left[\vol\left(\tilde{L}_{x, x'}\right)+\vol\left(\tilde{B}^c_x \cap \tilde{B}^c_{x'}\right)\right]^{m}\nonumber\\
    &=\left[1- \vol\left(\widetilde{C}_{x,x'}\right)\right]^{m}.\label{eq: one minus volume}
\end{align}


We will now bound the volume $\vol\left(\widetilde{C}_{x,x'}\right)$ for $x, x' \in \mathbb{T}^d$ using the condition $r<\frac{1}{4}$. 

If $d_t(x,x')\geq 2r$, then the two balls do not intersect, therefore $\vol\big(\tilde{C}_{x,x'}\big)= \vol\big(\tilde{B}_x\big) + \vol\big(\tilde{B}_{x'}\big)$. Note when the radius of a ball is less than a quarter, the wrapped around part will not overlap with the set. So, a linear translation of the center of the ball to  $(1/2,1/2,\ldots,1/2)$ will result in a ball with exactly the same volume as that of a Euclidean sphere in $\mathbb{R}^d$ of radius $r$. This implies that $\vol\big(\tilde{B}_x\big)=\vol\big(\tilde{B}_{x'}\big)= c_dr^d$, where $c_d$ denotes the volume of a Euclidean sphere of unit radius in $d$ dimensions, which implies that
\begin{align}\label{eq:no inter}
    \vol\big(\tilde{C}_{x,x'}\big) = 2c_dr^d \quad \text{if} \ d_t(x,x')\geq 2r.
\end{align}

On the other hand, in the case of $d_t(x,x')\leq 2r$, the two balls intersect. Also, note that $\tilde{B}_x \cup \tilde{B}_{x'}$ lies in a ball of radius $2r$, which is less than $1/2$. Again, by linearly translating the center of this ball to $(1/2,1/2,\ldots,1/2)$, we can see that the whole ball of radius $2r$ is strictly contained in $[0,1]^d$ without any wrapped-around part. As linear translation does affect the volume, we can conclude that if $d_t(x,x')\leq 2r$,
\begin{align}
     \vol\left(\tilde{C}_{x,x'}\right) &= \vol\left(C_{y,y'}\right), \label{eq:equal cres}
\end{align}
for some points $y,y' \in \mathbb{R}^d$ such that $d_t(x,x') = ||y-y'||$. The following lemma provides a bound on the volume $\vol\left(C_{y,y'}\right)$.

\begin{lemma}\label{lem:vol_cresc_bound}
If $||y-y'||<2r$, then $\vol\left(C_{y,y'}\right) \geq 2\cdot \frac{c_{d-1}r^{d-1}}{d}||y-y'||$, where $c_d$ denotes the volume of a Euclidean sphere of unit radius in $d$ dimensions.  
\end{lemma}
\begin{proof} The condition $||y-y'||<2r$ means that there is a non-empty intersection between the spheres. By symmetry and translational-invariance of the Lebesgue measure, the volume of the union of two crescents $C_{y,y'}$ is equal to twice the volume of the unshaded region of the sphere centered at origin shown on the right side of Fig.~\ref{fig:lens_volume}. The unshaded region is obtained by moving half of the lens to the opposite pole of the sphere, and it is the set of all points in the Euclidean unit sphere whose first coordinate lies between $-\frac{||y-y'||}{2}$ and $\frac{||y-y'||}{2}$.

The desired  volume in between the two shaded regions is lower bounded by the volume enclosed by the two cones whose bases are in the hyperplane $x_1=0$ with radii $R$ and the vertices are at points $\left(\frac{||y-y'||}{2}, 0,\ldots, 0\right)$ and $-\left(\frac{||y-y'||}{2}, 0,\ldots, 0\right)$. As the total volume of these cones is $\frac{c_{d-1}r^{d-1}}{d}||y-y'||$ for the contant $c_{d-1}$ is the volume of a Euclidean sphere of unit radius in $\mathbb{R}^{d-1}$, we have 
$$\frac{\vol\left(C_{y,y'}\right)}{2} \geq \frac{c_{d-1}r^{d-1}}{d}||y-y'||.$$ 

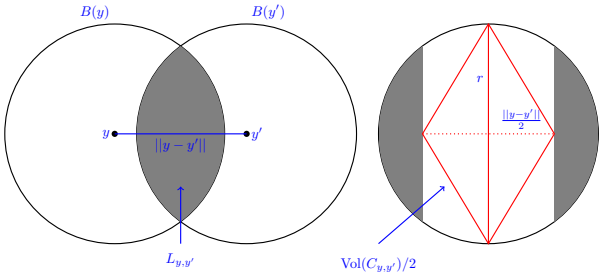
\begin{figure}[h!]
\centering
\resizebox{0.9\columnwidth}{!}{\begin{tikzpicture}[scale=1.1]

\begin{scope}
    \def\r{2.5}
    \def\d{3}

    \draw[thick] (-\d/2,0) circle (\r);
    \draw[thick] ( \d/2,0) circle (\r);

    \begin{scope}
        \clip (-\d/2,0) circle (\r);
        \fill[pattern=north east lines,gray]
            ( \d/2,0) circle (\r);
    \end{scope}

    \fill (-\d/2,0) circle (2pt);
    \fill ( \d/2,0) circle (2pt);

    \draw[blue,thick] (-\d/2,0) -- (\d/2,0);
    \node[blue,below] at (0,0) {$||y-y'||$};

    \node[blue,above left]  at (-\d/2,\r) {$B(y)$};
    \node[blue,above right] at ( \d/2,\r) {$B(y')$};

    \draw[blue,->,thick] (0,-2.5) -- (0,-1.2);
    \node[blue,below] at (0,-2.6) {$L_{y,y'}$};

    \node[blue,left]  at (-\d/2,0) {$y$};
    \node[blue,right] at ( \d/2,0) {$y'$};
\end{scope}

\begin{scope}[xshift=7cm]

    \def\r{2.5}    
    \def\a{3}      

    \draw[thick] (0,0) circle (\r);

    \begin{scope}
        \clip (0,0) circle (\r);

        \fill[pattern=north east lines,gray]
            (-\r,-\r) -- (-\a/2,-\r) -- (-\a/2,\r) -- (-\r,\r) -- cycle;

        \fill[pattern=north east lines,gray]
            (\r,-\r) -- (\a/2,-\r) -- (\a/2,\r) -- (\r,\r) -- cycle;
    \end{scope}

    \draw[red,thick] 
        (0,\r) -- (\a/2,0) -- (0,-\r) -- (-\a/2,0) -- cycle;

    \draw[red,thick] (0,-\r) -- (0,\r);
    \draw[red,thick,dotted] (-\a/2,0) -- (\a/2,0);

    \node[blue, left] at (0,\r/2) {$r$};
    \node[blue, above] at (\a/4,0) {$\frac{||y-y'||}{2}$};

    \draw[blue,->,thick] (-\r,-2.5) -- (-\a/2+0.5,-1.2);
    \node[blue,below] at (-\r,-2.7) {$\vol(C_{y,y'})/2$};

\end{scope}

\end{tikzpicture}}
\caption{The volume of the two crescents $C_{y,y'}$ is twice the volume of the unshaded region of the sphere centered at origin shown on the right side. The unshaded region is lower bounded by the volume enclosed by the two cones whose bases are in the hyperplane $y_1=0$ with radii $R$ and the vertices are at points $\left(-\frac{||y-y'||}{2}, 0,\ldots, 0\right)$ and $\left(\frac{||y-y'||}{2}, 0,\ldots, 0\right)$.} 
    \label{fig:lens_volume}
\end{figure}
\end{proof}

By combining Lemma~\ref{lem:vol_cresc_bound} and \eqref{eq:equal cres}, we can conclude that 
\begin{align}\label{eq:inter}
    \vol\big(\tilde{C}_{x,x'}\big) \geq  2\cdot \frac{c_{d-1}r^{d-1}}{d} d_t(x,x') \quad \text{if} \ d_t(x,x')\leq 2r.
\end{align}
Based on the distance $d_t(x,x')$, we can use \eqref{eq:no inter} or \eqref{eq:inter} to bound \eqref{eq: one minus volume}, and then we can  substitute back into \eqref{eq: vol_double_indicator} to obtain
\begin{align}
&\mathbb{E}\left[\vol(A_{E_{10},\ldots,E_{m0},X_1,\ldots,X_m})\right]\nonumber\\
&= \int \limits_x  \int \limits_{x'}  \sum_{e_{10}, \ldots, e_{m0}}\mathbb{E}\left[\mathbbm{1}_A(x) \cdot\mathbbm{1}_A(x')\right]\ dx'\  dx\nonumber\\
&\leq \int \limits_x  \int \limits_{x'}  \bigg\lbrace\left[1- \frac{2c_{d-1}r^{d-1}}{d}d_t(x,x')\right]^{m} \mathbbm{1}\{d_t(x,x')< 2r\} \nonumber\\
& \hspace{6.5 em}+ [1-2c_dr^d]^m\mathbbm{1}\{d_t(x,x')\geq 2r\}\bigg\rbrace\ dx'\  dx\nonumber\\
&= \int_0^{\sqrt{d}} \bigg\lbrace\left[1- \frac{2c_{d-1}r^{d-1}}{d}s\right]^{m} \mathbbm{1}\{s< 2r\} \nonumber\\
& \hspace{6.5 em} + [1-2c_dr^d]^m\mathbbm{1}\{s\geq 2r\}\bigg\rbrace f(s) ds,\label{eq:volume_bound}
\end{align}
where $f$ is the probability density function of the random distance $d_t(X,X')$, where $X$ and $X'$ are  independent and uniformly distributed points over $\mathbb{T}^d$. Note that as $r\leq \frac{1}{4}$, $2c_dr^d \leq 1$ and $\frac{2c_{d-1}r^{d-1}}{d}s \leq 1$ for $s\leq 2r$. The integral corresponding to the second term is upper bounded by $[1-2c_dr^d]^m$. For the first term, we need the following lemma.
\begin{lemma}\label{lem:gamma_integral}
Let $F$ and $f$ be the cumulative distribution function (CDF) and the probability density function (PDF), respectively, of $d_t(X, X')$, where $X$ and $X'$ are  independent and uniformly distributed points over $\mathbb{T}^d$. Let $K$ be a constant, and choose $s_0$ to be any value less than $\frac{1}{2}$ such that for all $s \leq s_0$, $0\leq (1-Ks)^m \leq 1$. Then, 
    \begin{align*}
        &\int \limits_0^{s_0} (1- Ks)^{m} f(s)\ ds \\& = (1- Ks_0)^{m} F(s_0)+m \int \limits_0^{s_0}(1- Ks)^{m-1} 
        F(s)\ ds\\
        & \leq (1- Ks_0)^{m} +\frac{c_d d!}{K^{d+1}}\cdot \frac{1}{m^d}.
    \end{align*}    
\end{lemma}

    \begin{proof}
        The equality follows by carrying out integration by parts. For the inequality, we can determine the value of CDF $F$ at $0 \leq s \leq s_0 \leq  \frac{1}{2} $:
        \begin{align*}
            F(s)&=\int \limits_x  \int \limits_{x'} \mathbbm{1}\{d_t(x,x') \leq s\} \ dx'\  dx \\&=\int \limits_x  \vol(\tilde{B}_x(s))\  dx
            \\&= \int \limits_x  c_d s^d  dx
            \\&= c_d s^d,
        \end{align*}
        where in the second equality, $\tilde{B}_x(s)$ is the ball of radius $s$ centered at $x$, and  the middle equality is due to the fact when the radius is smaller than $\frac{1}{2}$, $\vol(\tilde{B}_x(s)) = c_d s^d$, which follows from the argument near \eqref{eq:no inter}. 
     By making use of the inequality $(1-x)^m \leq e^{-mx}$ for $0 \leq x \leq 1$, we can write 
     \begin{align}
         m\int \limits_0^{s_0}(1- Ks)^{m-1} F(s)\ ds &\leq m\int \limits_0^{s_0} e^{-mKs} F(s)\ ds\nonumber\\
         &=mc_d \int \limits_0^{s_0} e^{-mKs} s^d\ ds\nonumber\\
         & = \frac{c_d}{K^{d+1}m^d}\int \limits_0^{\frac{s_0}{Km}} e^{-u} u^d\ du\nonumber\\
         & \leq \frac{c_d}{K^{d+1}m^d}\int \limits_0^{\infty} e^{-u} u^d\ du\nonumber\\
         &=\frac{c_d}{K^{d+1}m^d}\Gamma(d+1)\nonumber\\
         &= \frac{c_d d!}{K^{d+1}}\cdot \frac{1}{m^d}\nonumber.
     \end{align}
     proving the lemma. 
    \end{proof}

By applying Lemma~\ref{lem:gamma_integral} to the first term of \eqref{eq:volume_bound} with $K=\frac{2c_{d-1}r^{d-1}}{d}$ and $s_0=2r$, we obtain
\begin{align}
&\mathbb{E}\left[\vol(A_{E_{10},\ldots,E_{m0},X_1,\ldots,X_m})\right]\nonumber\\
&\leq \int_0^{\sqrt{d}} \bigg\lbrace\left[1- \frac{2c_{d-1}r^{d-1}}{d}s\right]^{m} \mathbbm{1}\{s< 2r\} \nonumber\\ & \hspace{8 em}+ [1-2c_dr^d]^m\mathbbm{1}\{s\geq 2r\}\bigg\rbrace f(s) ds,\nonumber\\
& \leq \left(1- \frac{2c_{d-1}r^{d-1}}{d}\cdot 2r\right)^{m} +\frac{c_d d^{d+1} d!}{(2c_{d-1}r^{d-1})^{d+1}}\cdot \frac{1}{m^d} \nonumber\\ & \hspace{17 em}+ (1-2c_dr^d)^m\nonumber\\
& \leq \alpha \cdot \frac{1}{m^d},
\end{align}
where the last inequality holds for all large enough $m$  for  some constant $\alpha>0$. This proves that
\begin{align}
H\big(E_{1 0},\ldots,E_{m 0}&\big| X_1, \ldots,X_{m}\big) \nonumber\\
    &\geq \log\frac{1}{\mathbb{E}\left[\vol(A_{E_{10},\ldots,E_{m0},X_1,\ldots,X_m}) \right]}\\
    & \geq \log \left(\frac{m^d}{\alpha}\right)= d\log m - \log(\alpha),
    \end{align}
    for all large enough $m$. By using this result in the argument of  \eqref{eq:entropy low arg}, we have the lower bound
$$H(G_m)\geq dm\log m - o(m \log m)$$
for the connection range $0 < r \leq \frac{1}{4}$, completing the proof of the theorem.

\end{document}